\pdfoutput=1

\documentclass[11pt, dvipsnames]{article}

\usepackage{natbib}

\usepackage{amsthm}
\usepackage{amsmath}
\usepackage{refcount}
\usepackage{fullpage}
\usepackage{paralist}
\usepackage{appendix}
\usepackage{graphicx}
\usepackage{bm}
\usepackage{bbold}
\usepackage{figs/tikzit} 

\tikzstyle{green dot}=[fill={rgb,255: red,0; green,120; blue,0}, draw=black, shape=circle]
\tikzstyle{red dot}=[fill=red, draw=black, shape=circle]
\tikzstyle{blue dot}=[fill=blue, draw=black, shape=circle]
\tikzstyle{purpule dot}=[fill={rgb,255: red,98; green,0; blue,98}, draw=black, shape=circle]
\tikzstyle{white dot}=[fill=white, draw=black, shape=circle]
\tikzstyle{yellow dot}=[fill=yellow, draw=black, shape=circle]
\tikzstyle{black sot}=[fill=black, draw=black, shape=circle]
\tikzstyle{white rectangle}=[fill=white, draw=black, shape=rectangle]
\tikzstyle{empty}=[fill=white, draw=white, shape=rectangle]

\tikzstyle{green edge}=[-, draw=green, fill=none]
\tikzstyle{red edge}=[-, draw=red]
\tikzstyle{blue edge}=[-, draw=blue]
\tikzstyle{purpule edge}=[-, draw={rgb,255: red,98; green,0; blue,98}]
\tikzstyle{special edge}=[draw=black, {|-|}]
\tikzstyle{special blue edge}=[draw=blue, {|-|}]
\tikzstyle{black edge}=[-, fill=none]

\usepackage{algpseudocode}
\usepackage{verbatim}

\renewcommand{\phi}{\varphi}

\DeclareMathOperator*{\argmin}{arg\,min}

\usepackage{todonotes}

\newcommand{\hide}[1]{ }

\renewcommand{\mathbf}{\bm}

\theoremstyle{plain}

\renewcommand{\include}{\input}

\usepackage{amsfonts}

\usepackage{tikz}
\usetikzlibrary{decorations.pathreplacing,angles,quotes}
\usetikzlibrary{arrows}
\usepackage{float}

\usetikzlibrary{patterns}

\usepackage{libertine}

\usepackage{geometry}
\usepackage{enumitem}
\PassOptionsToPackage{vlined, boxruled}{algorithm2e}
\usepackage{algorithm2e}
\LinesNumbered

\usepackage{mleftright}
\usepackage{relsize}

\usepackage{caption}
\usepackage{subcaption}

\usepackage{bbm}

\usepackage{amsmath}
\usepackage{amsthm}
\usepackage{graphicx}
\usepackage{tabularx}
\newcolumntype{Y}{>{\centering\arraybackslash}X}

\usepackage{array}
\usepackage{soul}

\usepackage{mdframed}
\usepackage{xparse}
\usepackage{hyperref}
\hypersetup{colorlinks,linkcolor=blue,citecolor=blue,urlcolor=blue}

\usepackage[capitalise, noabbrev]{cleveref}

\usepackage{tcolorbox}
\usepackage{tikz,lipsum}
\usepackage{newtxmath}
\tcbuselibrary{skins,breakable}

\theoremstyle{plain}
\newtheorem{thm1}{Theorem}[section]
\theoremstyle{remark}
\newtheorem{pthm1}[thm1]{Theorem}
\theoremstyle{plain}
\newtheorem{lem1}[thm1]{Lemma}
\theoremstyle{plain}
\newtheorem{obs1}[thm1]{Observation}
\theoremstyle{plain}
\newtheorem{inv1}[thm1]{Invariant}
\theoremstyle{plain}
\newtheorem{cor1}[thm1]{Corollary}
\theoremstyle{definition}
\newtheorem{defn1}[thm1]{Definition}
\theoremstyle{plain}
\newtheorem{fact1}[thm1]{Fact}
\theoremstyle{remark}
\newtheorem{rem1}[thm1]{Remark}
\theoremstyle{plain}
\newtheorem{prop1}[thm1]{Proposition}
\theoremstyle{plain}
\newtheorem{asmp1}[thm1]{Assumption}

\ifx\proof\undefined
\newenvironment{proof}[1][\protect\proofname]{\par
\normalfont\topsep6\p@\@plus6\p@\relax
\trivlist
\itemindent\parindent
\item[\hskip\labelsep\scshape #1]\ignorespaces
}{%
\endtrivlist\@endpefalse
}
\providecommand{\proofname}{Proof}
\fi

\def\special{1}
\def\specialproof{1}
\def\specialdefinition{1}
\def\specialremark{1}
 \def\highlight{0}
 \def\sidebar{1}

\def\colorequations{0}

\newenvironment{theorem}[1][]{%
\begin{thm1}[#1]%
}{\end{thm1}%
}
\newenvironment{lemma}[1][]{%
\begin{lem1}[#1]%
}{\end{lem1}%
}
\newenvironment{observation}[1][]{%
\begin{obs1}[#1]%
}{\end{obs1}%
}
\newenvironment{inv}[1][]{%
\begin{inv1}[#1]%
}{\end{inv1}%
}
\newenvironment{fact}[1][]{%
\begin{fact1}[#1]%
}{\end{fact1}%
}
\newenvironment{remark}[1][]{%
\begin{rem1}[#1]%
}{\end{rem1}%
}
\newenvironment{pthm}[1][]{%
\begin{pthm1}[#1]%
}{\end{pthm1}%
}
\newenvironment{corollary}[1][]{%
\begin{cor1}[#1]%
}{\end{cor1}%
}
\newenvironment{definition}[1][]{%
\begin{defn1}[#1]%
}{\end{defn1}%
}
\newenvironment{asmp}[1][]{%
\begin{asmp1}[#1]%
}{\end{asmp1}%
}
\newenvironment{proposition}[1][]{%
\begin{prop1}[#1]%
}{\end{prop1}
}%

\if 1
    \if 1\highlight
        \renewenvironment{theorem}[1][]{%
        \begin{mdframed}[nobreak=false,backgroundcolor=Aquamarine!60]\begin{thm1}[#1]%
        }{\end{thm1}\end{mdframed}%
        }
        \renewenvironment{lemma}[1][]{%
        \begin{mdframed}[nobreak=false,backgroundcolor=YellowGreen!60]\begin{lem1}[#1]%
        }{\end{lem1}\end{mdframed}%
        }
        \renewenvironment{observation}[1][]{%
        \begin{mdframed}[nobreak=false,backgroundcolor=Salmon!60]\begin{obs1}[#1]%
        }{\end{obs1}\end{mdframed}%
        }

        \renewenvironment{corollary}[1][]{%
        \begin{mdframed}[backgroundcolor=Mulberry!60]\begin{cor1}[#1]%
        }{\end{cor1}\end{mdframed}%
        }
    \fi
    \if 1\sidebar
        \tcolorboxenvironment{theorem}{blanker,breakable,left=4mm, before skip=10pt,after skip=10pt, borderline west={2mm}{0pt}{Aquamarine}}
        \tcolorboxenvironment{lemma}{blanker,breakable,left=4mm, before skip=10pt,after skip=10pt, borderline west={2mm}{0pt}{YellowGreen}}
        \tcolorboxenvironment{observation}{blanker,breakable,left=4mm, before skip=10pt,after skip=10pt, borderline west={2mm}{0pt}{Salmon}}
        \tcolorboxenvironment{inv}{blanker,breakable,left=4mm, before skip=10pt,after skip=10pt, borderline west={2mm}{0pt}{Salmon}}
        \tcolorboxenvironment{fact}{blanker,breakable,left=4mm, before skip=10pt,after skip=10pt, borderline west={2mm}{0pt}{Salmon}}
        \tcolorboxenvironment{pthm}{blanker,breakable,left=4mm, before skip=10pt,after skip=10pt, borderline west={2mm}{0pt}{Salmon}}
        \tcolorboxenvironment{proposition}{blanker,breakable,left=4mm, before skip=10pt,after skip=10pt, borderline west={2mm}{0pt}{Goldenrod}}
        \tcolorboxenvironment{corollary}{blanker,breakable,left=4mm, before skip=10pt,after skip=10pt, borderline west={2mm}{0pt}{Mulberry}}
        \tcolorboxenvironment{asmp}{blanker,breakable,left=4mm, before skip=10pt,after skip=10pt, borderline west={2mm}{0pt}{Goldenrod}}
    \fi
\fi

\if 1\specialproof
    \if 1\highlight
        \expandafter\let\expandafter\oldproof\csname\string\proof\endcsname
        \let\oldendproof\endproof
        \renewenvironment{proof}[1][\proofname]{%
        \begin{mdframed}[nobreak=false,backgroundcolor=lightgray!60]\oldproof[#1]%
        }{\oldendproof\end{mdframed}}
    \else
        \if 1\sidebar
        \tcolorboxenvironment{proof}{
        blanker,breakable,left=4mm,
        before skip=10pt,after skip=10pt,
        borderline west={2mm}{0pt}{lightgray}}
        \fi
    \fi
\fi

\if 1\specialdefinition
    \if 1\highlight
        \renewenvironment{definition}[1][]{%
        \begin{mdframed}[innerbottommargin=0.1cm,innertopmargin=0.1cm,backgroundcolor=Apricot!60]\begin{defn1}[#1]%
        }{\end{defn1}\end{mdframed}%
        }
    \else
        \if 1\sidebar
        \tcolorboxenvironment{definition}{
        blanker,breakable,left=4mm,
        before skip=10pt,after skip=10pt,
        borderline west={2mm}{0pt}{Apricot}}
        \fi
    \fi
\fi

\if 1\specialremark
    \if 1\highlight
        
    \else
        \if 1\sidebar
        \tcolorboxenvironment{remark}{
        blanker,breakable,left=4mm,
        before skip=10pt,after skip=10pt,
        borderline west={2mm}{0pt}{Salmon}}
        \fi
    \fi
\fi

\if 1\colorequations
    \everymath{\color{MidnightBlue}}
\fi

\title{Online Joint Replenishment Problem with Arbitrary Holding and Backlog Costs}
 \author{Yossi Azar
     \thanks{School of Computer Science, Tel-Aviv University, Israel. Email: azar@tau.ac.il. Research supported in part by the Israel Science Foundation (grant No. 2304/20).}
 	\and
 	Shahar Lewkowicz
 	\thanks{School of Computer Science, Tel-Aviv University, Israel. Email: shaharlewko22@gmail.com.}
 }

\begin{document}

\maketitle

\begin{abstract}

In their seminal paper Moseley, Niaparast, and Ravi \cite{doi:10.1137/1.9781611978322.130} introduced the Joint Replenishment Problem (JRP) with holding and backlog costs that models the
trade-off between ordering costs, holding costs, and backlog costs in supply chain planning systems. Their model generalized the classical the make-to-order version as well make-to-stock version. For the case where holding costs function of all items are the same and all backlog costs are the same, they provide a constant competitive algorithm, leaving designing a constant competitive algorithm for arbitrary functions open. Moreover, they noticed that their algorithm does not work for arbitrary (request dependent) holding costs and backlog costs functions. We resolve their open problem and design a constant competitive algorithm that works for arbitrary request dependent functions. 
Specifically, we establish a 4-competitive algorithm for the single-item case and a 16-competitive for the general (multi-item) version.
The algorithm of \cite{doi:10.1137/1.9781611978322.130} is based on fixed priority on the requests to items, and request to an item are always served by order of deadlines. In contrast, we design an algorithm with dynamic priority over the requests such that instead of servicing a prefix by deadline of requests, we may need to service a general subset of the requests.

\end{abstract}

\thispagestyle{empty}
\newpage
\setcounter{page}{1}
\section{Introduction}
Joint replenishment (JRP) is one of the most fundamental problems in  supply chain management
\cite{literature,PengWangWang}. It has been considered in various papers both for the offline model and for the online model. Over time, there is a demand for items that should be satisfied by ordering the  items. Ordering items has an ordering cost. Specifically, every order irrespective of the set
of items ordered incurs a joint ordering cost $c(r)$. In addition, if an item $v$ is part of the order, there is an item-dependent cost of $c(v)$ irrespective of the number of units of item $v$ ordered. Traditionally, there are two versions. The first is that requests for items should be satisfied upon item arrival and holding items encounter holding costs. Then the goal is to minimize the holding cost in additional to the ordering cost. In the second version, holding inventory is not allowed and only backlogging that encounter a delay cost is possible. Then the goal to minimize the backlog or delay cost in additional to the ordering cost. The first version is called the make-to-stock version, while the second one is called make-to-order version.

In the online version, Bienkowski et al. \cite{bienkowski2014better} designed a
2-competitive algorithm for the make-to-stock version and showed that this competitive ratio was optimal. For the  make-to-order version, Buchbinder et al. \cite{buchbinder2013online} gave an elegant 3-competitive algorithm using the online primal-dual framework. They also provide a lower bound of approximately 2.64 on the competitiveness of any deterministic online
algorithm for the problem. 

Recently, in a seminal paper, Moseley, Niaparast, and Ravi \cite{doi:10.1137/1.9781611978322.130} considered the model that both holding and backlog are allowed. Specifically, each request arrives at a certain time with a target time (also called deadline). If the request is served before its deadline, the algorithms suffers from holding cost that is non-increasing up to the deadline of the requests, and if it is served after the deadline it suffers from the backlog or delay costs which is non-decreasing over the time after the deadline. They considered both the single item and the multi-item versions and provided a constant competitive algorithm for the problems by assuming that all functions are identical. 
The case in which each request has its own arbitrary holding and backlog functions 
has remained opened. Moreover, they provide a lower bound for their algorithm (even for a single item) when the functions are  different. They concluded that their algorithm fails to be constant competitive for general holding and backlog functions. In addition, their proof techniques are somewhat based on linear programming.

We resolved this open problem and provide a constant competitive algorithm for the JRP problem with arbitrary and possibly different functions for each request, both for the single item version and for the multi-item. Our algorithm is based on dynamic priority of the requests. The dynamic priority is based on virtual deadlines, which are changing over time for each request. This is in contrast to the algorithm of \cite{doi:10.1137/1.9781611978322.130} which is based on static priority.
Specifically, a request with an earlier deadline always has a higher priority to be served over a request with a later one. However, for our algorithm the priority between requests is changing over time (in principle unlimited number of times) and decisions are done in a more delicate way.

\subsection{Our Results}
 Our main results are as follows:

\begin{theorem}
There exists a 4-competitive deterministic polynomial-time algorithm for Online JRP with
backlog and holding costs with a single item for arbitrary (and request-dependent) monotone backlog and holding functions.
\end{theorem}
We then extend this to the general multi-item case with a worse competitive ratio.

\begin{theorem}
There exists a 16-competitive deterministic polynomial-time algorithm for Online multi-item JRP with backlog and holding costs for arbitrary (and request-dependent) monotone backlog and holding functions.
\end{theorem}

The algorithms use some basic ingredients that were used in previous algorithms but has some major differences since previous algorithms for this problem are known to be non-constant competitive. For a single item, the algorithm decides to serve a request once the total backlog cost of requests (that reached their deadlines) reaches the ordering cost $c(r)$. Then, previous algorithms (for identical functions) always prefer to serve requests with earlier deadlines and the only question is where  to cut the prefix. In contrast to that, our algorithm has a dynamic priority over the requests that may change an arbitrary number of times, and hence the subset of served request may be arbitrary. Specifically, we define a virtual deadline for each request that changes over time (while the deadline remains fixed). In particular, the virtual deadline of the requests are decreasing over time in a nonuniform request-dependent way. The priority of which requests to serve depends on these dynamic virtual deadlines.

For the multi-item case, the algorithm is more complicated. First, services are triggered when the backlog delay of requests on a subset of elements becomes large enough. Then the triggered elements are the ones that suffer a total backlog delay which is larger than the item cost $c(v)$. All pending requests in their backlog period are served. Moreover, at this opportunity more items should be served (Premature Backlog Phase) and one needs to decide which. In previous algorithms, the idea was to forward the time and see which elements will be triggered first (and cut the prefix at a given cost). In contrast, an important ingredient that is crucial to our algorithm is not to let all requests contributed equally. Specifically, when we forward the time from the current time $t_1$ to a time $t_2$, the contribution of each request to trigger an element at $t_2$ is done only by requests whose current virtual deadline (at time $t_1$) is smaller then $t_2$.
The above delicate definition is captured by the variable $\hat{mature}$ which is different from the classical and standard variable $mature$ that has been used previously. Finally, we serve requests of triggered items in their holding period similarly to a single item both Local Holding Phase (per triggered item separately) and in Global Holding Phase (globally among all triggered items).

\subsection{Previous Work}

We start with the online version of the problem.
As mentioned before for make-to-stock JRP version, Bienkowski et al. \cite{bienkowski2014better}
designed a 2-competitive online algorithm and prove that it has an optimal competitive ratio.
For online make-to-order JRP Brito et al. \cite{brito2012competitive} gave a 5-competitive algorithm. Then, Buchbinder et al. \cite{buchbinder2013online} gave a 3-competitive algorithm for this problem and proved a lower bound of 2.64 for this problem. The lower bound was later improved to 2.754 in \cite{bienkowski2014better}. These lower bounds apply even for the case of linear delay functions.

For the online version of the single-item make-to-order JRP (which is also called Dynamic TCP acknowledgment), Dooly et al \cite{dooly1998tcp,dooly2001line} designed a 2-competitive algorithm, and showed  a matching lower bound for any deterministic algorithm. 
Seiden \cite{seiden2000guessing} showed a $e/(e-1)$ lower bound in the randomized online algorithm. Later, Karlin et al. \cite{karlin2001dynamic} designed a randomized algorithm which achieves the optimal $e/(e-1)$ competitive ratio.

As discussed before, the only paper that considered both holding cost and backlog cost is the seminal paper of Moseley, Niaparast, and Ravi \cite{doi:10.1137/1.9781611978322.130}. However, it assumes that all holding function of all requests are identical, as well as that all backlog delay functions of all requests are identical.

The offline versions of the two variant of JRP (make-to-stock and make-to-order) are known to be strongly NP-hard \cite{arkin1989computational,becchetti2009latency,nonner2009approximating}.
There is a sequence of papers that provide below two approximation starting with 
Levi, Roundy and Shmoys \cite{LRS06} who used a primal-dual approach for both holding and backlog costs. Then Levi et al. \cite{levi2008constant,levi2006improved}  improved the bound to 1.8 and later Bienkowski et al. \cite{bienkowski2014better} further improved the approximation ratio to 1.791. 

There are also results for special cases: for zero holding costs, the ratio was reduced to $5/3$ by Nonner and Souza \cite{nonner2009approximating} and then to 1.574 by Bienkowski et al. \cite{bienkowski2015approximation}.
The special case of the make-to-order JRP with a single item (Dynamic TCP Acknowledgement) can be solved with dynamic programming\cite{dooly1998tcp,dooly2001line}.

A more general problem to the JRP is the Online Aggregation problems \cite{chrobak2014online}, which was considered in several papers. Here we are given multi-level tree 
(JRP can be viewed as two-level tree). Requests appears on the leaves and services is done with rooted subtree. This problem has been considered only with backlog (waiting) cost. 
Bienkowski et al. \cite{bienkowski2020online} provided a $O(D^4 2^D)$-competitive algorithm where $D$ is the depth of the tree. 
This ratio is later improved to $O(D^2)$ by Azar and Touitou \cite{azar2019general}. 
For the special case of a path, Bienkowski et al. \cite{bienkowski2013online}
showed that the competitive ratio of this problem is between 3.618 and 5.
There are many more additional results (see \cite{bienkowski2021new,bosman2020improved,buchbinder2017depth,cheung2016submodular,FNR14,mari2024online,mcmahan2021d,nagarajan2016approximation}).

A more general problem is called Online Service with Delay (OSD), first introduced by Azar et al. \cite{azar2021online} on arbitrary metric spaces. In this problem, 
a server can move in the graph to serve the requests and the goal is to minimize the sum of the distance traveled by the server and delay cost. They gave $O(\log^4 n)$-competitive algorithm where n is the number of nodes in the graph.
This result was subsequently improved to $O(\log^2 n)$ by Azar and Touitou \cite{azar2019general} and later by Touitou \cite{touitou2023improved} to $O(\log n)$ which is the best currently known competitive ratio for this problem.

\section{Preliminaries}
\label{section_preliminaries}


We denote by $ALG$ the online algorithm we are currently analyzing, which is Algorithm \ref{hold_back_single_item.alg} in section \ref{section_single_item} (which deals with the single item problem) and Algorithm \ref{hold_back_multi_item.alg} in section \ref{section_multi_item} (which deals with the multi items problem). We denote by $OPT$ the optimal algorithm for the problem.

In the single item problem the service cost is denoted by $c$.

In the multi items problem we denote by $V=\{v_1,v_2,...,v_n\}$ the set of items. The cost of serving the item $1\leq i\leq n$ is $c(v_i)$ and the cost of the root is $c(r)$. For a service $S$, regardless if it was triggered by $ALG$ or $OPT$, we denote by $V(S)\subseteq V$ the set of items that were served in the service $S$. Therefore the service cost of the service $S$ is $c(r)+\sum_{v\in V(S)}c(v)$.

In the single item problem each request $p$ is defined by $(a_p,d_p,c_p)$ and in the multi items problem each request $p$ is defined by $p=(v_p,a_p,d_p,c_p)$ where:
\begin{itemize}
\item $v_p\in V$ is the item being requested.

\item $a_p$ is the arrival time of the request $p$, i.e. this is the time when the online algorithm is exposed to the request $p$. In addition, both $ALG$ and $OPT$ cannot serve the request $p$ prior to this time. 

\item $d_p$ is the deadline of the request $p$ (where $a_p\leq d_p$). 

\item $c_p:[a_p,\infty)\rightarrow R_{\geq 0}$ is the cost of $p$. $c_p(t)$ is the cost of the request $p$ at time $t$. The cost function $c_p$ is monotone-non increasing in the range $[a_p,d_p]$, in which we refer to it as holding cost, and is monotone-non decreasing in the range $[d_p,\infty)$, in which we refer to it as backlog cost.
Without loss of generality, we assume that each request $p$ is exposed at time $a_p$ (if it is exposed before, we can ignore it until time $a_p$).

\end{itemize}

We define the following:
\begin{definition}
For a request $p$ and time $t$ such that $t\leq d_p$ we define the \textbf{virtual deadline of $p$ with respect to $t$}, which we denote as $\hat{d_p^t}$, as the time after the deadline of $p$ where its backlog cost will reach the holding cost at time $t$, i.e. $\hat{d_p^t}\geq d_p$ is defined such that $c_p(\hat{d_p^t})=c_p(t)$.
\end{definition}

An example of the term virtual deadline appears in figure \ref{fig.virtualDeadlineFig}. Note that since both holding cost and backlog cost functions are monotone, we may assume that they are continuous (otherwise perturb them and make then continuous).

\begin{figure}
\centering
\includegraphics[width= 8cm]{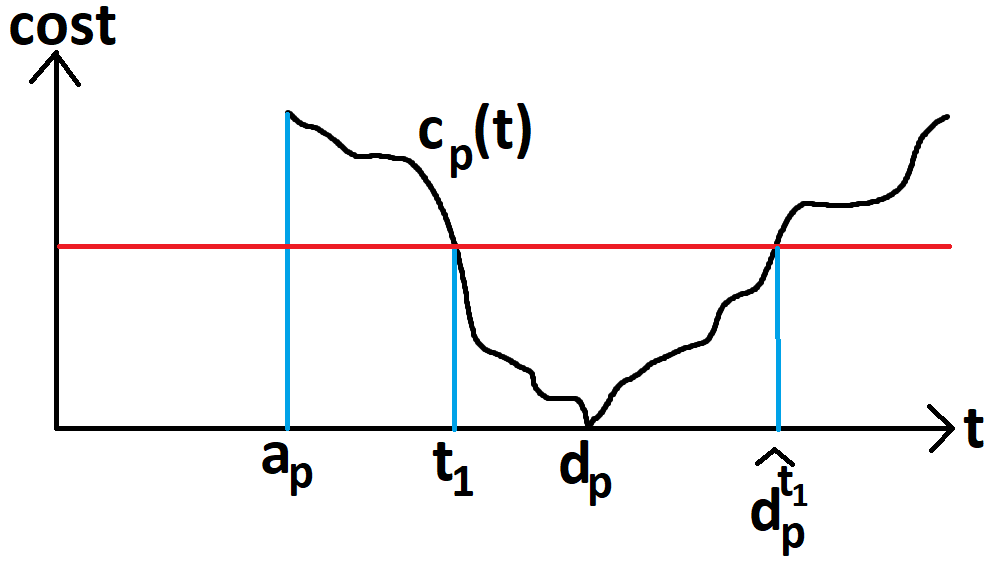}
\caption{The virtual deadline of the request $p$ at time $t_1$ is the time $\hat{d_p^{t_1}}$. The red horizontal line is the line $cost=c_p(t_1)$, or $cost=c_p(\hat{d_p^{t_1}})$, which is the same.}
\label{fig.virtualDeadlineFig}
\end{figure}

For a request $p$ we denote by $t_p^{ALG}$ the time when $ALG$ served $p$. Recall that $t_p^{ALG}\geq a_p$. We denote by $c_p^{ALG}$ the holding cost that $ALG$ paid for serving $p$ if $t_p^{ALG}\leq d_p$ and the backlog cost that $ALG$ paid for serving $p$ if $t_p^{ALG}\geq d_p$. In the same way we define $t_p^{OPT}$ (where $t_p^{OPT}\geq a_p)$ and $c_p^{OPT}$. It can be seen that we have $$c_p^{ALG}=c_p(t_p^{ALG}),$$ $$  c_p^{OPT}=c_p(t_p^{OPT}).$$

We denote by $\Psi_{ALG}$ ($\Psi_{OPT}$) the set of services that $ALG$ ($OPT$) have triggered when analyzing a (known from context) sequence of requests $\sigma$. 
For an item $v\in V$, we denote by $\Psi_{OPT}^{v}$ all the services of $OPT$ which included the item $v$, i.e. $$\Psi_{OPT}^v=\{S\in\Psi_{OPT}:v\in V(S)\}.$$

We denote by $ALG(\sigma)$ ($OPT(\sigma)$) the cost $ALG$ ($OPT$) pay for serving the sequence $\sigma$ and have that $ALG$ is $\alpha$-competitive if for every sequence of requests $\sigma$ we have $$ALG(\sigma)\leq \alpha\cdot OPT(\sigma).$$ Recall that the cost of an algorithm in the single item problem is the number of services it triggers multiplied by $c$, in addition to the holding and backlog costs the algorithm pays for the requests. In the multi items problem the cost of an algorithm is the total costs of the services it makes, where the cost of a service $S$ is $c(r)+\sum_{v\in V(S)}c(v)$, in addition to the holding and backlog costs the algorithm pays for the requests.

We have that $OPT$ does not trigger more than one service at the same time (otherwise combine the services and reduce the cost, contradicting the fact that $OPT$ is an optimal algorithm). The algorithms we show for the single item problem and for the multi items problem also have the property that they do not trigger more than one service at the same time. 

Therefore, in both the single item and multi items, we can consider the sequence of services that both $ALG$ and $OPT$ performed (ordered according to the time they performed them). If both $ALG$ and $OPT$ triggered one service (each) at the same time, then their order in the sequence can be chosen arbitrary.

\section{Holding and Backlog Costs - Single Item}
\label{section_single_item}

\subsection{Algorithm}

To describe our algorithm, we need some definitions. Recall that for time $t$ we define the "Virtual deadline" of request $p$ with respect to time $t\leq d_p$ as the time $\hat{d_p^t}\geq d_p$ for which $c_p(\hat{d_p^t})=c_p(t)$. Note that the virtual deadline is changing over time. Moreover, the order of the virtual deadlines of different requests may change over time.

Similar to previous algorithms, a service is triggered once the sum of the backlog costs becomes $c$. The main difference between our algorithm and the algorithm of \cite{doi:10.1137/1.9781611978322.130} is that in the holding phase we serve additional requests according to their virtual deadlines with respect to the current time $t$ when the service is being triggered rather than the actual deadlines. We also ignore requests which cost more than $c$ in the holding phase. 
In each service, we include additional requests as long as the sum of their holding costs is at most $2\cdot c$.
However, the analysis is very different because the usage of linear programs does not seem to be effective when dealing with arbitrary holding and backlog costs.

\begin{algorithm}[H]
\caption{Single-Item Online Joint Replenishment Problem}
\label{hold_back_single_item.alg}
\textbf{Event triggering.} Let $t$ be the time in which the sum of the backlog costs of unsatisfied overdue requests (denote by $T$) becomes $c$ \textbf{do}:

\Indp 

\textbf{Backlog Phase.} Trigger a service $S$ at time $t$ and satisfy all requests in $T$.

\textbf{Holding phase.} Consider all the remaining active unsatisfied requests with costs of at most $c$ at time $t$ in order of their \textbf{virtual} deadlines with respect to $t$. Add them one by one to $S$, as long as the sum of their holding costs at time $t$ is at most $2\cdot c$.

\Indm
\end{algorithm}

Throughout this section, Algorithm \ref{hold_back_single_item.alg} will be denoted by $ALG$.

\subsection{Analysis}

Henceforth, in this section our target is to prove the following theorem.

\begin{theorem}
\label{hold_back_single_item.theorem}
Algorithm \ref{hold_back_single_item.alg} is $4$-competitive.
\end{theorem}

Consider the sequence $\sigma$ of requests and the sequence of services that both $ALG$ and $OPT$ triggered, ordered according to the time they were triggered.
For the purpose of the following definition, we add one dummy service of $OPT$ that does not serve any request after all the services of $ALG$ and $OPT$ have been triggered. We do not charge $OPT$ for the cost of that dummy service. We partition $\Psi_{ALG}$ to the following four disjoint sets $\Psi_A,\Psi_B,\Psi_C,\Psi_D$ as follows:





\begin{definition}
\label{definition_partition_psi_alg_single_item} Firstly,
\begin{itemize}
\item Let $\Psi_A$ be all the services in $\Psi_{ALG}$ that occurred before the first service that $OPT$ performed.
\end{itemize}
Given two consecutive services $S_i,S_{i+1}\in \Psi_{OPT}$ triggered at times $t_i$ and $t_{i+1}$ (where $t_i < t_{i+1}$) let $J_i$ be the first service that $ALG$ performed between $S_i$ and $S_{i+1}$ (triggered at time $y_i$) such that $J_i$ served at least one request in the holding phase with virtual deadline with respect to time $y_i$ of at least $t_{i+1}$ \textbf{or} the sum of all the holding costs of the requests $J_i$ served in its holding phase was of at most $c$. 
\begin{itemize}
\item Let $\Psi_B$ be union over $i$ of all the services of $ALG$ between $S_i$ and $J_i$ (excluding).
    If the service $J_i$ as defined above does not exist, then all the services of $ALG$ between $S_i$ and $S_{i+1}$ are in $\Psi_B$.
\item Let $\Psi_C$ be union over $i$ of all the services $J_i$ .
\item Let $\Psi_D$ be union over $i$ of all the services of $ALG$ between $J_i$ and $S_{i+1}$ (excluding).
\end{itemize}
\end{definition}

An example of a possible partition appears in figure \ref{fig.psiABCDsingleItem}.

Note that $ALG(\sigma) = c\cdot|\Psi_{ALG}|+\sum_{p}c_p^{ALG}$ and $OPT(\sigma) = c\cdot|\Psi_{OPT}|+\sum_{p}c_p^{OPT}$.

Observe that for each service $S\in\Psi_{ALG}$ we have that $ALG$ pays $c$ service cost, $c$ backlog cost and at most $2c$ holding cost, which sums up to a total cost of at most $4c$. This yields the following observation.
\begin{observation}
\label{hold_back.lemma.service_cost_alg}
    The total cost $ALG$ paid is at most $4$ times the number of services it made: $$ALG(\sigma)\leq 4\cdot c\cdot |\Psi_{ALG}|.$$
\end{observation}
Note that we can consider a function that maps each service $S\in\Psi_C$ to the last service of $OPT$ that occurred before it, and this function is single-valued. Also note that the dummy service of $OPT$ which we added after all the services of $ALG$ and $OPT$ in definition \ref{definition_partition_psi_alg_single_item} does not have any service of $\Psi_C$ that is mapped to it since it was "triggered" by $OPT$ after all the services of $ALG$ and $OPT$. An illustration appears in Figure \ref{fig.psiABCDsingleItem}. This implies the following observation:
\begin{observation}
\label{hold_back.lemma.type_c}
    We have that $$|\Psi_C|\leq |\Psi_{OPT}|.$$
\end{observation}
\begin{figure}[t]
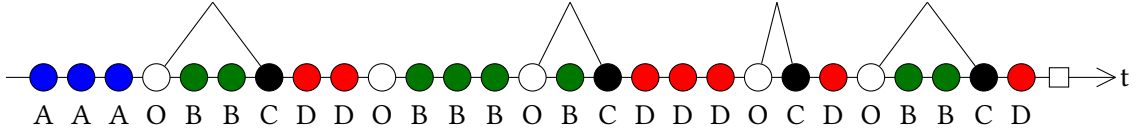

\centering
\ctikzfig{figs/psiABCDsingleItem}
\caption{An example of a possible sequence of services of $ALG$ and $OPT$ during the time horizon. The white dots represent services of $OPT$. The white rectangle in the right represents the dummy service of $OPT$ from definition \ref{definition_partition_psi_alg_single_item}. The other dots represent services of $ALG$: the blue represent services in $\Psi_A$, the green represent services in $\Psi_B$, the black represent services in $\Psi_C$ and the red represent services in $\Psi_D$. The edges represents the one-way function from $\Psi_C$ to $\Psi_{OPT}$ which is the reason for Observation \ref{hold_back.lemma.type_c}. Note that between the second and third services of $OPT$ there are only services of $\Psi_B$, as a result the second service of $OPT$ does not have any service in $\Psi_C$ that the function mentioned above maps to it.}
\label{fig.psiABCDsingleItem}
\end{figure}
Next, we show that for each of the services $S\in\Psi_{ALG}\setminus\Psi_C$ that the total cost $OPT$ pays for the requests $ALG$ served in $S$ is at least $c$. This, combined with observations \ref{hold_back.lemma.type_c} and \ref{hold_back.lemma.service_cost_alg} would imply a competitive ratio of $4$.
We claim the following:
\begin{lemma}
\label{hold_back.lemma.type}
    For each service $S\in\Psi_{ALG}\setminus\Psi_C$ the total cost $OPT$ pays for the requests $ALG$ served in the service $S$ is at least $c$.
\end{lemma}

Before proving Lemma \ref{hold_back.lemma.type} we show how to use it to prove Theorem \ref{hold_back_single_item.theorem}.

\begin{proof}[Proof of Theorem \ref{hold_back_single_item.theorem}]
Clearly, each request $p$ is served in at most one service in $\Psi_{ALG}\setminus\Psi_C$. By using lemma \ref{hold_back.lemma.type}
and summing for all the services in $\Psi_{ALG}\setminus\Psi_C$ we get that $$c\cdot |\Psi_{ALG}\setminus\Psi_C|\leq \sum_{p}c_p^{OPT}.$$
The above inequality together with observation \ref{hold_back.lemma.type_c} yield
\begin{align*}
OPT(\sigma) & = c\cdot|\Psi_{OPT}|+\sum_{p}c_p^{OPT} 
 \geq c\cdot|\Psi_C|+c\cdot|\Psi_{ALG}\setminus\Psi_C|=c\cdot |\Psi_{ALG}|
\end{align*}
and thus
$$ALG(\sigma)\leq 4\cdot c\cdot |\Psi_{ALG}|\leq 4\cdot OPT(\sigma)$$ where the first inequality is due to observation \ref{hold_back.lemma.service_cost_alg}.
\end{proof}

\subsection{Analyzing the services of type A}
\label{section_holdBack_single_item_a}

The goal of this subsection is to prove lemma \ref{hold_back.lemma.type} for the case $S\in\Psi_A$. First, we have the following lemma:

\begin{lemma}
\label{hold_back.lemma.type_a_single_request}
Let $S\in\Psi_A$ be a service and let $p$ be a request which $ALG$ served in $S$ during the backlog phase. Then $$c_p^{ALG}\leq c_p^{OPT}.$$
\end{lemma}
\begin{proof}
    We have $d_p\leq t_p^{ALG} \leq t_p^{OPT}$ where the first inequality follows since the request $p$ was served by $ALG$ during the backlog phase of the service $S$, and the second inequality is since $ALG$ triggered $S$ before the first time $OPT$ triggered a service (due to $S\in\Psi_A$) so, in particular, $ALG$ served $p$ before $OPT$ served $p$. Due to the monotonicity of the backlog cost we have $c_p^{ALG}\leq c_p^{OPT}.$
\end{proof}

We now prove lemma \ref{hold_back.lemma.type} for the case $S\in\Psi_A$.

\begin{proof}[Proof of Lemma \ref{hold_back.lemma.type} (case $S\in\Psi_A$)]
Note that $ALG$ pays a total backlog cost of $c$ for the requests it served during the backlog phase of the service $S$. The lemma follows by summing for these requests and using lemma \ref{hold_back.lemma.type_a_single_request}.
\end{proof}

\subsection{Analyzing the services of type B}
\label{section_holdBack_single_item_b}

The goal of this subsection is to prove lemma \ref{hold_back.lemma.type} for the case $S\in\Psi_B$. In order to do that, we begin by presenting an observation which follows from the definition of $\Psi_B$. Let $t_0$ be the time $ALG$ triggered the service $S$ and let $t_1$ be the time of the next service triggered by $OPT$ after the service $S$. 

\begin{observation}
\label{hold_back.lemma.type_b_properties}
For each service $S\in\Psi_B$ we have:
\begin{enumerate}
\item For each request $p$ that was served during the holding phase of $S$ we have $\hat{d_p^{t_0}}\leq t_1$.
\item The sum of all the holding costs of the requests $S$ served in its holding phase was at least $c$.
\end{enumerate}
\end{observation}

We show the following lemma:

\begin{lemma}
\label{hold_back.lemma.type_b_single_request}
Let $S\in\Psi_B$ be a service and let $p$ be a request which $ALG$ served in $S$ during the holding phase. Then, $$c_p^{ALG}\leq c_p^{OPT}.$$
\end{lemma}
\begin{proof}
As denoted above $t_0=t_p^{ALG}$. We have $t_0\leq d_p$ since $p$ was served by $ALG$ using the service $S$ during the holding phase.
In case $OPT$ served the request $p$ before $ALG$, i.e. $t_p^{OPT}\leq t_0$ we have that $t_p^{OPT}\leq t_0\leq d_p$. The monotonicity of the holding cost implies that $c_p^{ALG}\leq c_p^{OPT}$.
The remaining case is that $OPT$ served the request $p$ after $ALG$, i.e. $t_p^{OPT}>t_0$. Recall that $t_1$ be the time of the next service triggered by $OPT$ after the service $S$. Note that the case of $t_p^{OPT}>t_0$ cannot happen if the next service triggered by $OPT$ after the service $S$ is the dummy service, 
since that dummy service served no requests. 

We have $$t_0\leq d_p\leq \hat{d_p^{t_0}} \leq t_1\leq t_p^{OPT}$$ where the third inequality holds due to observation \ref{hold_back.lemma.type_b_properties} (part 1) and the last inequality is due to the definition of $t_1$. Therefore, $$c_p^{ALG}=c_p(t_0)=c_p(\hat{d_p^{t_0}})\leq c_p(t_p^{OPT})=c_p^{OPT}$$ where the second equality is due to the definition of virtual deadlines and the first inequality is due to the monotonicity of the backlog cost.
\end{proof}

We now prove lemma \ref{hold_back.lemma.type} for the case $S\in\Psi_B$.

\begin{proof}[Proof of Lemma \ref{hold_back.lemma.type} (case $S\in\Psi_B$)]
Due to observation \ref{hold_back.lemma.type_b_properties} (part 2) we have that $ALG$ pays a total cost of at least $c$ for the requests it served in the service $S$ during the holding phase. The lemma follows by summing for these requests together with lemma \ref{hold_back.lemma.type_b_single_request}.
\end{proof}

\subsection{Analyzing the services of type D}
\label{section_holdBack_single_item_d}

The goal of this subsection is to prove lemma \ref{hold_back.lemma.type} for the case $S\in\Psi_D$. First, we have the following lemma:

\begin{lemma}
\label{hold_back.lemma.type_d_single_request}
Let $S\in\Psi_D$ be a service and let $p$ be a request which $ALG$ served in $S$ during the backlog phase such that $c_p^{OPT}\leq c$. Then $$c_p^{ALG}\leq c_p^{OPT}.$$
\end{lemma}
\begin{proof}

We have $d_p\leq t_p^{ALG}$ since the request $p$ was served by $ALG$ during the backlog phase of the service $S$.
In case $OPT$ served the request $p$ after $ALG$ served it, i.e. $t_p^{ALG}\leq t_p^{OPT}$ we have $d_p\leq t_p^{ALG}\leq t_p^{OPT}$ and thus due to the monotonicity of the backlog cost we have $c_p^{ALG}\leq c_p^{OPT}.$ We note that the case where $OPT$ served the request $p$ after $ALG$ served it does not happen if the next service triggered by $OPT$ after the service $S$ is the dummy service 
since that dummy service served no requests.

Hence, we are left with the case where $OPT$ served the request $p$ before $ALG$ served $p$, i.e. $t_p^{OPT}< t_p^{ALG}$.

Let $S_C\in\Psi_C$ be the furthest service in $\Psi_C$ that was triggered by $ALG$ before the service $S$ and let $t_0$ be the time $ALG$ triggered the service $S_C$. Since $S\in\Psi_D$ we know that there were not any services triggered by $OPT$ between the two services $S_C$ and $S$, therefore we must have that $OPT$ served the request $p$ before the service $S_C$, i.e. $t_p^{OPT}\leq t_0$.

Recall that $t_p^{ALG}$ is the time when $ALG$ triggered the service $S$ (and served the request $p$). Let $t_1$ be the time of the next service performed by $OPT$ after the service $S$. Such a service performed by $OPT$ must exist since $S\in\Psi_D$. Therefore, $t_p^{ALG}\leq t_1$.

Recall that $OPT$ served the request $p$ before the service $S_C$ so $p$ was active in $ALG$ at time $t_0$. Hence, $t_0\leq d_p$ since otherwise $ALG$ would have served the request $p$ during the backlog phase of service $S_C$ rather than during the backlog phase of the service $S$.

Summing the above, we have 
\begin{equation} \label{inequality_times_d_single_item}
t_p^{OPT}\leq t_0\leq d_p\leq t_p^{ALG}\leq t_1 . 
\end{equation}

Next, we show that $t_1 \leq \hat{d_p^{t_0}}$.
Due to the monotonicity of the holding costs we have $$c_p(t_0)\leq c_p(t_p^{OPT})=c_p^{OPT}\leq c.$$ Observe that since $S_C\in\Psi_C$ we have that the service $S_C$ satisfies at least one of the following two conditions:
\begin{enumerate}
    \item There was a request $q$ served by $ALG$ during the holding phase of the service $S_C$ with a virtual deadline with respect to time $t_0$ of at least $t_1$, i.e. $\hat{d_q^{t_0}}\geq t_1$.
    \item The sum of the holding costs of the requests $S_C$ served in its holding phase was at most $c$.
\end{enumerate}
If the service $S_C$ satisfied the second condition and the request $p$ was active in $ALG$ at time $t_0$ and $c_p(t_0)\leq c$ then $ALG$ could have included the request $p$ in the holding phase of the service $S_C$ and the total holding costs of the requests $S_C$ served in its holding phase would be still at most $2c$. Hence, that yields a contradiction. 
Hence, $ALG$ did not include the request $p$ during the holding phase of service $S_C$ since the first condition occur. 
Recall that, $ALG$ considers which requests to serve in $S_C$ according to their virtual deadlines with respect to time $t_0$. Since $ALG$ includes the request $q$ but not the request $p$ in the holding phase of service $S_C$, we have that $t_1\leq \hat{d_q^{t_0}}\leq \hat{d_p^{t_0}}$ and hence indeed
$t_1\leq \hat{d_p^{t_0}}$. The last inequality, together with 
inequality \ref{inequality_times_d_single_item} implies that 
$t_p^{ALG}\leq t_1 \leq \hat{d_p^{t_0}}$.

Hence, 
$$c_p^{ALG}=c_p(t_p^{ALG})\leq c_p(\hat{d_p^{t_0}})=c_p(t_0)\leq f(t_p^{OPT})=c_p^{OPT}$$ where the first inequality is due to the inequality above and the monotonicity of the backlog costs, the second equality follows from the definition of virtual deadlines and the second inequality holds due to inequality \ref{inequality_times_d_single_item} and the monotonicity of the holding cost.

\end{proof}

We now prove lemma \ref{hold_back.lemma.type} for the case $S\in\Psi_D$.

\begin{proof}[Proof of Lemma \ref{hold_back.lemma.type} (case $S\in\Psi_D$)]
If there is a request $p$ which $ALG$ served in the service $S$ such that $c_p^{OPT}\geq c$ then we are done. Hence, the remaining case is that for each request $p$ that $ALG$ served in the service $S$ we have $c_p^{OPT}<c$.
Note that $ALG$ pays a total backlog cost of $c$ for the requests it served during the backlog phase of the service $S$. The lemma follows by summing for these requests and using lemma \ref{hold_back.lemma.type_d_single_request}.
\end{proof}
\section{Holding and Backlog Costs - Multi Items}
\label{section_multi_item}

If the joint service cost \( c(r) \) exceeds the total cost of individual items, i.e., $c(r) > \sum_{i=1}^n c(v_i),$
the problem simplifies to a single-item Joint Replenishment Problem (JRP) with service cost \( c(r) \), since the combined cost of servicing all items is within a constant factor of \( c(r) \). 
Conversely, if there exists an item \( v \) such that 
$ c(r) < c(v), $
we can treat \( v \) as a separate single-item problem, incurring only an additive constant  overhead to the competitive ratio. 
Although no assumptions are made, these observations suggest that the most challenging scenario arises when for all $i$ we have
$ c(r) > c(v_i)$
yet the total $\sum_{i=1}^n c(v_i)$
is significantly larger than \( c(r) \). 
In such multi-item cases, the algorithm must justify incurring both the joint service cost and the individual item costs for each service.

\subsection{Algorithm}

Let $A_t$ be the set of all the active requests at time $t$, i.e. the requests that have arrived before or at time $t$ and are not satisfied by any of the previous services of the algorithm.

For an item $v\in V$, let $A^v_t$ be the set of all active requests at time $t$ for the item $v$, i.e. $$A^v_t=\{p\in A_t:v_p=v\}.$$

The algorithm maintains a variable $b(v)$ for each item $v$ which starts with a value of $0$ at time $t=0$ and for each time $t$ its value $b_t(v)$ indicates the amount of backlog accumulated by the unsatisfied overdue requests for $v$ at time $t$, i.e. $b_t(v)=\sum_{p\in A^v_t,d_p<t} c_p(t)$. We call an item $v$ mature if $b_t(v)\geq c(v)$. Any backlog cost that is accumulated after $v$ becomes mature is called surplus backlog cost.

The algorithm also maintains a variable $b(r)$ to the root, which for each time $t$ its value $b_t(r)$ indicates the sum of the surplus backlog costs for the mature items at time $t$, i.e. $b_t(r)=\sum_{v:b_t(v)\geq c(v)}(b_t(v)-c(v))$.

For each premature item $v$ that has at least one unsatisfied request at time $t$, the authors of \cite{doi:10.1137/1.9781611978322.130} defined another term, $mature_t(v)$ to be the time $t'>t$ in which the variable $b(v)$ will reach the value of $c(v)$, i.e. $\sum_{p\in A^v_t,d_p<t'}c_p(t')=c(v).$ This is the time when the item $v$ will become mature. \textbf{We define another term}: $\hat{mature}_t(v)$. Unlike \cite{doi:10.1137/1.9781611978322.130}, we want to take into account only requests that either reached their deadline at time less or equal to $t$ or requests with deadline later than $t$ that will have a bigger backlog cost at the time when the variable is filled than their holding cost at the current time $t$. In other words, requests $p$ with $d_p>t$ are ignored until the time when their backlog cost reaches their current holding cost at time $t$, rather than their (earlier) deadlines $d_p$, only then we take into account their contribution to $b(v)$ when calculating $\hat{mature}_t(v)$. Formally: 
\begin{definition}
For each time $t$ and each premature item $v\in V$ of time $t$ that has at least one unsatisfied request at time $t$ we define 
$$\hat{mature}_t(v)=\argmin_{t'}\left\{\left(\sum_{p\in A^v_t,d_p\leq t}c_p(t')+\sum_{p\in A^v_t,t<d_p\leq t',c_p(t)\leq c_p(t')}c_p(t')\right)\geq c(v)\right\}.$$ 
\end{definition}
Please note that unlike the definition of $mature_t(v)$ as given by \cite{doi:10.1137/1.9781611978322.130}, unless the item $v$ is served shortly, the variable $b(v)$ will probably reach the value of $c(v)$ before (rather than at) the time $\hat{mature}_t(v)$ because we do not take into account requests which have high holding cost at time $t$.

Another difference between our algorithm and the algorithm of \cite{doi:10.1137/1.9781611978322.130} is that in the local holding phase and global holding phase we serve requests in order according to their virtual deadline with respect to the service time rather than according to their deadlines. This is similar to the difference between our algorithm and their algorithm for single item as well. In the local (global) holding phase, we ignore requests that have a cost of more than $c(v)$ ($c(r)$) at the time of the service.

\begin{algorithm}[H]
\caption{Multi-Item Online Joint Replenishment Problem}
\label{hold_back_multi_item.alg}

\textbf{Event triggering.} Let $t$ be the time in which $b(r)$ reaches the value of $c(r)$ (i.e. the surplus backlog cost accumulated equals the joint service cost) \textbf{do}:

\Indp 

\textbf{Mature Backlog Phase.} Trigger a service $S$ at time $t$. Include all the mature items $v$ in $S$, satisfy all their overdue requests using $S$ and remove them from $A_t$.

\textbf{Premature Backlog Phase.} Sort all the premature items in non-decreasing order of $\hat{mature}_t(v)$ and include them in $S$ one by one, as long as the sum of their costs $c(v)$ is at most $2\cdot c(r)$. For each newly included item, satisfy all of their overdue requests using $S$ and remove them from $A_t$.

\textbf{Local Holding Phase.} For each item $v$ included in $S$, iterate over the remaining active requests for $v$ with costs of at most $c(v)$ at time $t$ in non-decreasing order of their \textbf{virtual} deadlines with respect to time $t$ and satisfy them one by one using $S$, as long as the sum of their holding cost is at most $2\cdot c(v)$. 

\textbf{Global Holding Phase.} Go over the remaining active requests for items included in $S$ with costs of at most $c(r)$ at time $t$ in non-decreasing order of their \textbf{virtual} deadlines with respect to time $t$ and satisfy them one by one using the service $S$, as long as the sum of their holding cost is at most $2\cdot c(r)$.

\Indm
\end{algorithm}

Throughout this section, Algorithm \ref{hold_back_multi_item.alg} will be denoted by $ALG$.


\subsection{Analysis}

Henceforth, in this section our target is to prove the following theorem.

\begin{theorem}
\label{hold_back_multi_item.theorem}
Algorithm \ref{hold_back_multi_item.alg} is $16$-competitive.
\end{theorem}

Recall that for a service $S$, regardless if it was triggered by $ALG$ or $OPT$, we denote by $V(S)\subseteq V$ the set of items that were served in the service $S$.

We denote by $P(S)$ the set of requests that were served during the service $S$.

If $S\in\Psi_{ALG}$ then we denote by $V_1(S)$ and $V_2(S)$ the set of items that were included in the service $S$ during the mature backlog phase and the premature backlog phases, respectively. Therefore we have $V(S)=V_1(S)\cup V_2(S)$ and $V_1(S)\cap V_2(S)=\emptyset$.

Fix $\sigma$, the sequence of requests. We begin by partitioning the services of $ALG$ to different types, as we did in the analysis of single item. 
As for the analysis of the single item problem, we add one dummy service of $OPT$ for all items that does not serve any request after all the services of $ALG$ and $OPT$ have been triggered. We do not charge $OPT$ for the cost of that dummy service.

We partition $\Psi_{ALG}$ to six disjoint sets. The first four sets are $\Psi_A,\Psi_B,\Psi_C,\Psi_D$ which are slightly different from single item and additional two new sets $\Psi_G,\Psi_H$ as follows. 

\begin{definition}
\label{definition_partition_psi_alg_multi_item} 
Initially, we define $R=\Psi_{ALG}$.
\begin{itemize}
\item 
Let $\Psi_H$ be all the services $S\in R$ that have $\sum_{v\in V_1(S)}c(v)\geq c(r)$. Remove $\Psi_H$ from $R$.
\item 
Let $\Psi_A$ be all the services $S\in R$ that occurred before the first service that $OPT$ performed. Remove $\Psi_A$ from $R$.
\end{itemize}
Let $S_i,S_{i+1}\in \Psi_{OPT}$ be two consecutive services triggered at times $t_i$ and $t_{i+1}$ (where $t_i < t_{i+1}$).   
\begin{itemize}
\item
Let $\Psi_G$ be the union over $i$ of services $S\in R$ triggered at time $z_i$ between $S_i$ and $S_{i+1}$ such that \textbf{both} $\sum_{v\in V_2(S)}c(v)\geq c(r)$ \textbf{and} $\max_{v\in V_2(S)}\hat{mature}_{z_i}(v)\leq t_{i+1}$. Remove $\Psi_G$ from $R$.
\item Let $J_i \in R$ be the first service of $R$ between $S_i$ and $S_{i+1}$ (triggered at time $y_i$) such that: $J_i$ served at least one request in the global holding phase with virtual deadline with respect to time $y_i$ of at least $t_{i+1}$ \textbf{or} the sum of all the holding costs of the requests $J_i$ served in its global holding phase was of at most $c(r)$.
\item
Let $\Psi_B$ be union over $i$ of all the services of $R$ between $S_i$ and $J_i$ (excluding). If the service $J_i$ as defined above does not exist, then all the services of $R$ between $S_i$ and $S_{i+1}$ will be in $\Psi_B$. 
\item
Let $\Psi_C$ be union over $i$ of $J_i$.
\item 
Let $\Psi_D$ be union over $i$ of all services of $R$ between $J_i$ and $S_{i+1}$ (excluding). 
\end{itemize}
\end{definition}

\begin{definition}
For each item $v\in V$ we denote by $\Psi_G^v$ the services $S\in\Psi_G$ that have the item $v$ served in the premature backlog phase. We also denote by $\Psi_H^v$ the services $S\in\Psi_H$ that have the item $v$ served in the mature backlog phase. More precisely we have
$$\Psi_G^v=\{S\in\Psi_G:v\in V_2(S)\}$$ $$\Psi_H^v=\{S\in\Psi_H:v\in V_1(S)\}$$
We also define
$$\Psi_{ABD}=\Psi_A\cup\Psi_B\cup\Psi_D.$$
\end{definition}

We now have the following formula for the cost of $OPT$ that consists of 3 terms: serving the root, serving the items and the costs of the requests (backlog and holding costs) .
\begin{observation}
\label{hold_back_multi.lemma.cost_opt}
We have $$OPT(\sigma)=|\Psi_{OPT}|\cdot c(r) + \sum_{v\in V}|\Psi_{OPT}^{v}|\cdot c(v) +\sum_{p}c_p^{OPT}$$
\end{observation}

We now have the following lemma regarding the cost of $ALG$:
\begin{lemma}
\label{hold_back_multi.lemma.cost_alg}
We have $$
ALG(\sigma) \leq 16\cdot \left( c(r)\cdot |\Psi_{ABD}|+ c(r)\cdot |\Psi_{C}|+\sum_{S\in\Psi_G}\sum_{v\in V_2(S)}c(v)+\sum_{S\in\Psi_H}\sum_{v\in V_1(S)}c(v)\right). $$
\end{lemma}
\begin{proof}
Let $S\in\Psi_{ALG}$ be a service of $ALG$. The service $S$ causes $ALG$ to pay a service cost of $c(r)$ for the service, a service cost of $\sum_{v\in V_1(S)}c(v)$ for the mature backlog phase and a service cost of $\sum_{v\in V_2(S)}c(v)$ for the premature backlog phase. The service $S$ also causes $ALG$ paying a surplus backlog cost of $c(r)$ for the requests for mature items, a non-surplus backlog cost of $\sum_{v\in V_1(S)}c(v)$ for the requests for mature items and a non-surplus backlog cost of at most $\sum_{v\in V_2(S)}c(v)$ for the requests for premature items. The local holding phase causes $ALG$ paying a holding cost of at most $2\cdot(\sum_{v\in V_1(S)}c(v)+\sum_{v\in V_2(S)}c(v))$. The global holding phase causes $ALG$ paying a holding cost of at most $2\cdot c(r)$. In total, each service $S$ contributed to the total cost of $ALG$ a cost of at most 
\begin{equation}
\label{multi_item_contribution_of_service_to_alg_cost}
    4\cdot\left(c(r)+\sum_{v\in V_1(S)}c(v)+\sum_{v\in V_2(S)}c(v)\right).
\end{equation}

If $S\in\Psi_{ABD}$ or $S\in\Psi_{C}$ then we have $\sum_{v\in V_1(S)}c(v)\leq c(r)$ and $\sum_{v\in V_2(S)}c(v)\leq 2\cdot c(r)$ and thus by plugging this in \ref{multi_item_contribution_of_service_to_alg_cost} we get that such a service contributed to the total cost of $ALG$ a cost of at most $16\cdot c(r)$.

If $S\in\Psi_G$ then we have $$\sum_{v\in V_1(S)}c(v)\leq c(r)\leq \sum_{v\in V_2(S)}c(v)$$ and thus by plugging this in \ref{multi_item_contribution_of_service_to_alg_cost} we get that such a service contributed to the total cost of $ALG$ a cost of at most $16\cdot \sum_{v\in V_2(S)}c(v)$.

If $S\in\Psi_H$ then we have $\sum_{v\in V_2(S)}c(v)\leq 2\cdot c(r)$ and $c(r)\leq \sum_{v\in V_1(S)}c(v)$, thus by plugging this in \ref{multi_item_contribution_of_service_to_alg_cost} we get that such a service contributed to the total cost of $ALG$ a cost of at most $16\cdot \sum_{v\in V_1(S)}c(v)$.
\end{proof}

Note that as in the analysis for the single item problem, we can consider a function that maps each service $S\in\Psi_C$ to the last service of $OPT$ that occurred before it, this function is single-valued and the dummy service of $OPT$ from definition \ref{definition_partition_psi_alg_multi_item} does not have any service in $\Psi_C$ that is mapped to it, as in the analysis of the single item problem. This yields the following observation:

\begin{observation}
\label{multi_item_hold_back.observation.type_c}
    We have that $$|\Psi_C|\leq |\Psi_{OPT}|.$$
\end{observation}

We claim the following two lemmas:
\begin{lemma}
\label{hold_back_multi_item.lemma.type_abd}
    For each service $S\in\Psi_{ABD}$ we have that the total cost $OPT$ pays for the requests $ALG$ served in $S$ is at least $c(r)$, i.e. $$c(r)\leq\sum_{p\in P(S)}c_p^{OPT}.$$
\end{lemma}

\begin{lemma}
\label{hold_back_multi_item.lemma.type_gh}
For each $v\in V$ we have 
$$c(v)\cdot(|\Psi_G^{v}|+|\Psi_H^{v}|) \leq c(v)\cdot |\Psi_{OPT}^{v}| +\sum_{S\in \Psi_G^{v}\cup \Psi_H^{v}}\sum_{p\in P^{v}(S)}c_p^{OPT}.$$
\end{lemma}
Before proving Lemma \ref{hold_back_multi_item.lemma.type_abd} and Lemma \ref{hold_back_multi_item.lemma.type_gh} 
we show how to prove the main theorem from the Lemmas.

\begin{proof}[Proof of Theorem \ref{hold_back_multi_item.theorem}]
Using Lemma \ref{hold_back_multi_item.lemma.type_abd} and summing it for all the services $S\in\Psi_{ABD}$ yields to
\begin{align*}
 c(r)\cdot |\Psi_{ABD}|
 \leq \sum_{S\in\Psi_{ABD}}\sum_{p\in P(S)}c_p^{OPT}.
\end{align*}
We also have 
\begin{align*}
\sum_{S\in\Psi_G}\sum_{v\in V_2(S)}c(v)+\sum_{S\in\Psi_H}\sum_{v\in V_1(S)}c(v) 
& =\sum_{v\in V}\sum_{S\in \Psi_G^{v}}c(v)+\sum_{v\in V}\sum_{S\in \Psi_H^{v}}c(v) \\
& =\sum_{v\in V}c(v)\cdot(|\Psi_G^{v}|+|\Psi_H^{v}|) \\
& \leq \sum_{v\in V}c(v)\cdot |\Psi_{OPT}^{v}| +\sum_{S\in \Psi_G^{v}\cup \Psi_H^{v}}\sum_{p\in P^{v}(S)}c_p^{OPT}
\end{align*}
where the first equality follows by changing the order of the summing and the inequality follows from applying Lemma \ref{hold_back_multi_item.lemma.type_gh} for each $v\in V$. Therefore, we have
\begin{align*}
& c(r)\cdot |\Psi_{ABD}|+ c(r)\cdot |\Psi_{C}|+\sum_{S\in\Psi_G}\sum_{v\in V_2(S)}c(v)+\sum_{S\in\Psi_H}\sum_{v\in V_1(S)}c(v) \\
\leq & \sum_{S\in\Psi_{ABD}}\sum_{p\in P(S)}c_p^{OPT}+c(r)\cdot |\Psi_{OPT}|+\sum_{v\in V}c(v)\cdot |\Psi_{OPT}^{v}| +\sum_{S\in \Psi_G^{v}\cup \Psi_H^{v}}\sum_{p\in P^{v}(S)}c_p^{OPT}\\
= & |\Psi_{OPT}|\cdot c(r)+\sum_{v\in V}c(v)\cdot |\Psi_{OPT}^{v}|+\sum_{S\in\Psi_{ABD}}\sum_{p\in P(S)}c_p^{OPT}+\sum_{v\in V}\sum_{S\in \Psi_G^{v}\cup \Psi_H^{v}}\sum_{p\in P^{v}(S)}c_p^{OPT} \\
\leq & |\Psi_{OPT}|\cdot c(r)+\sum_{v\in V}c(v)\cdot |\Psi_{OPT}^{v}|+\sum_{p}c_p^{OPT}
\end{align*}
where in the first inequality we used Observation \ref{multi_item_hold_back.observation.type_c} and the second inequality holds due to the fact that each request $p$ is served by $ALG$ in at most one service. Hence, we have 
\begin{align*}
ALG(\sigma) & \leq  16\cdot \left( c(r)\cdot |\Psi_{ABD}|+ c(r)\cdot |\Psi_{C}|+\sum_{S\in\Psi_G}\sum_{v\in V_2(S)}c(v)+\sum_{S\in\Psi_H}\sum_{v\in V_1(S)}c(v)\right) \\
& \leq 16\cdot \left( |\Psi_{OPT}|\cdot c(r)+\sum_{v\in V}c(v)\cdot |\Psi_{OPT}^{v}|+\sum_{p}c_p^{OPT} \right) = 16\cdot OPT(\sigma).
\end{align*}
where the first inequality holds due to Lemma \ref{hold_back_multi.lemma.cost_alg} and the equality follows from Observation \ref{hold_back_multi.lemma.cost_opt}.
\end{proof}

For a service $S\in\Psi_{ALG}$ we denote by $B_1(S)$, $B_2(S)$, $H_1(S)$ and $H_2(S)$ as the set of requests satisfied using the service $S$ in the mature backlog phase, the premature backlog phase, the local holding phase and the global holding phase, respectively.

For an item $v\in V(S)$, we define $B_1^v(S)$, $B_2^v(S)$, $H_1^v(S)$ and $H_2^v(S)$ in the same way as above except that we include only requests for the item $v$ in the definition. 

We define $P^v(S)$ to be the set of requests for the item $v$ that were served during the service $S$ in any of the phases.

Therefore we have $B_1(S)=\bigcup_{v\in V_1(S)}B_1^v(S)$, $B_2(S)=\bigcup_{v\in V_2(S)}B_2^v(S)$, $H_1(S)=\bigcup_{v\in V(S)}H_1^v(S)$ and $H_2(S)=\bigcup_{v\in V(S)}H_2^v(S)$.

We have for each item $v\in V_1(S)$ that $P^v(S)=B_1^v(S)\cup H_1^v(S) \cup H_2^v(S)$ and for each item $v\in V_2(S)$ that $P^v(S)=B_2^v(S)\cup H_1^v(S) \cup H_2^v(S).$ We also have $\bigcup_{v\in V(S)}P^v(S)=P(S).$

\subsection{Analyzing the services of type A}
\label{section_multi_item_holdBack_a}

The goal of this subsection is to prove lemma \ref{hold_back_multi_item.lemma.type_abd} for the case $S\in\Psi_A$. First, we have the following lemma:

\begin{lemma}
\label{hold_back_multi_item.lemma.type_a_single_request}
Let $S\in\Psi_A$ be a service and let $p\in B_1(S)$ be a request which $ALG$ served in $S$ during the mature backlog phase. Then $$c_p^{ALG}\leq c_p^{OPT}.$$
\end{lemma}
\begin{proof}
The proof follows the proof of Lemma \ref{hold_back.lemma.type_a_single_request} in subsection \ref{section_holdBack_single_item_a} where instead of referring to the backlog phase of Algorithm \ref{hold_back_single_item.alg}, we refer to the mature backlog phase of Algorithm \ref{hold_back_multi_item.alg}.
\end{proof}

We now prove lemma \ref{hold_back_multi_item.lemma.type_abd} for the case $S\in\Psi_A$.

\begin{proof}[Proof of Lemma \ref{hold_back_multi_item.lemma.type_abd} (case $S\in\Psi_A$)]
We have $$c(r)\leq c(r)+\sum_{v\in V_1(S)}c(v)=\sum_{p\in B_1(S)}c_p^{ALG}\leq \sum_{p\in B_1(S)}c_p^{OPT}\leq \sum_{p\in P(S)}c_p^{OPT}$$ where the equality is since the service $S$ was triggered when the variable $b(r)$, which equals to the sum of the surplus backlog costs for the mature items, reached the value of $c(r)$, and the variable $b(v)$ for each item $v\in B_1(S)$ had a value of $c(v)$, due to the item $v$ being mature. The second inequality follows by summing for the requests $p\in B_1(S)$ and using lemma \ref{hold_back_multi_item.lemma.type_a_single_request}.
\end{proof}

\subsection{Analyzing the services of type B}
\label{section_multi_item_holdBack_b}

The goal of this subsection is to prove lemma \ref{hold_back_multi_item.lemma.type_abd} for the case $S\in\Psi_B$. We begin by presenting this observation which follows from the definition of $\Psi_B$.
Let $t_0$ be the time $ALG$ triggered the service $S$ and let $t_1$ be the time of the next service triggered by $OPT$ after the service $S$.

\begin{observation}
\label{hold_back_multi_item.lemma.type_b_properties}
For each service $S\in\Psi_B$ we have:
\begin{enumerate}
\item  For each request $p$ that was served during the global holding phase of $S$ we have $\hat{d_p^{t_0}}\leq t_1$.
\item We have $$c(r)\leq \sum_{p\in H_2(S)}c_p^{ALG}.$$
\end{enumerate}
\end{observation}

Note that Observation \ref{hold_back_multi_item.lemma.type_b_properties} is equivalent to Observation \ref{hold_back.lemma.type_b_properties} where the cost of $c$ is replaced with $c(r)$ and instead of referring to the holding phase of Algorithm \ref{hold_back_single_item.alg}, we refer to the global holding phase of Algorithm \ref{hold_back_multi_item.alg}. This yields the following lemma:

\begin{lemma}
\label{hold_back_multi_item.lemma.type_b_single_request}
Let $S\in\Psi_B$ be a service and let $p$ be a request which $ALG$ served in $S$ during the global holding phase. Then $$c_p^{ALG}\leq c_p^{OPT}.$$
\end{lemma}
\begin{proof}
The proof follows the proof of Lemma \ref{hold_back.lemma.type_b_single_request} in subsection \ref{section_holdBack_single_item_b} where instead of referring to the holding phase of Algorithm \ref{hold_back_single_item.alg}, we refer to the global holding phase of Algorithm \ref{hold_back_multi_item.alg}. Also, in the proof, Observation \ref{hold_back_multi_item.lemma.type_b_properties} should be used rather than Observation \ref{hold_back.lemma.type_b_properties}.
\end{proof}

We now prove lemma \ref{hold_back_multi_item.lemma.type_abd} for the case $S\in\Psi_B$.

\begin{proof}[Proof of Lemma \ref{hold_back_multi_item.lemma.type_abd} (case $S\in\Psi_B$)]
We have $$c(r)\leq \sum_{p\in H_2(S)}c_p^{ALG}\leq \sum_{p\in H_2(S)}c_p^{OPT}\leq \sum_{p\in P(S)}c_p^{OPT}$$

where the first inequality is due to Observation \ref{hold_back_multi_item.lemma.type_b_properties} (part 2) and the second inequality follows by summing for the requests in $H_2(S)$ and using lemma \ref{hold_back_multi_item.lemma.type_b_single_request}.
\end{proof}

\subsection{Analyzing the services of type D}
\label{section_multi_item_holdBack_d}

The target of this subsection is to prove lemma \ref{hold_back_multi_item.lemma.type_abd} for the case $S\in\Psi_D$. Let $S\in\Psi_D$ be a service that $ALG$ triggered, and let $t_1$ be the time when $ALG$ triggered $S$. Let $S_C\in\Psi_C$ be the furthest service in $\Psi_C$ that was triggered by $ALG$ before the service $S$ and let $t_0$ be the time when $ALG$ triggered the service $S_C$. Since $S\in\Psi_D$ there were not any services triggered by $OPT$ between the two services $S_C$ and $S$. In order to prove Lemma \ref{hold_back_multi_item.lemma.type_abd} for the case $S\in\Psi_D$, we will analyze the requests that were served during the mature backlog phase of the service $S$. We begin with the requests for items that were served in the service $S_C$ (this is done in the following lemma) and then we will consider the  requests for items that were not served in the service $S_C$.
\begin{lemma} \label{hold_back_multi_item_lemma_tyep_d_yes_s_c}
Let $v\in V_1(S)\cap V(S_C)$ be an item that was included in the mature backlog phase of the service $S$ and was also served in the service $S_C$. Let $p\in B_1^v(S)$ be a request for the item $v$ that was served during the mature backlog phase of $S$. If we have $c_p^{OPT}\leq c(r)$ then we have $$c_p^{ALG}\leq c_p^{OPT}.$$
\end{lemma}
\begin{proof}
    The proof follows the proof of Lemma \ref{hold_back.lemma.type_d_single_request} in subsection \ref{section_holdBack_single_item_d} where instead of referring to the backlog phase of Algorithm \ref{hold_back_single_item.alg}, we refer to the mature backlog phase of Algorithm \ref{hold_back_multi_item.alg}, instead of referring to the holding phase of Algorithm \ref{hold_back_single_item.alg}, we refer to the global holding phase of Algorithm \ref{hold_back_multi_item.alg}, and we also use $c=c(r)$. Note that the fact that Algorithm \ref{hold_back_multi_item.alg} also performs the premature backlog phase and local holding phase does not affect the proof.
\end{proof}
Now we analyze the items that were included in the mature backlog phase of the service $S$ and were not served in the service $S_C$. Let $v\in V_1(S)\setminus V(S_C)$ be such an item. 

Recall that each request $p\in B_1^v(S)$ for the item $v$ that was satisfied by the service $S$ in the mature backlog phase has $d_p\leq t_1$ since $p$ was served during the mature backlog phase of the service $S$. We partition the set $B_1^v(S)$ to the following three disjoint sets:
\begin{itemize}
\item $P_1$: the requests $p\in B_1^v(S)$ with $a_p\leq t_0$ \textbf{and} ( $d_p\leq t_0$ \textbf{or} $c_p(t_0)\leq c_p(t_1)$).
\item $P_2$: the requests $p\in B_1^v(S)$ with $a_p\leq t_0$ \textbf{and} $d_p> t_0$ \textbf{and} $c_p(t_0)>c_p(t_1)$.
\item $P_3$: the requests $p\in B_1^v(S)$ for which $t_0<a_p\leq t_1$. 
\end{itemize}
Observe that the requests in $P_1$ and $P_2$ have arrived before $ALG$ triggered the service $S_C$. The requests of $P_3$ have arrived after $ALG$ triggered the service $S_C$ (but before $ALG$ triggered the service $S$). Also observe that $B_1^v(S)=P_1\cup P_2\cup P_3$ and that these three sets are disjoint.

The following three lemmas deal with the requests $P_1$, $P_2$ and $P_3$ respectively.
\begin{lemma}
\label{hold_back_multi_item.lemma.type_d_p_1}
We have $$\sum_{p\in P_1}c_p^{ALG}< c(v).$$
\end{lemma}
\begin{proof}
Let $t_2$ be the time of the next service that $OPT$ triggered after $ALG$ triggered the service $S$. Recall that $ALG$ triggered the service $S$ at time $t_1$ and thus $$t_0< t_1\leq t_2.$$ Since $S_C\in\Psi_C$ rather than $S_C\in\Psi_G$ at least one of the following two cases occurred:

\begin{enumerate}
\item We have $$\sum_{v'\in V_2(S_C)}c(v')<c(r).$$
\item There is an item $v'\in V_2(S_C)$ such that $t_2<\hat{mature}_{t_0}(v').$
\end{enumerate}
Assume by contradiction that $$c(v)\leq\sum_{p\in P_1}c_p^{ALG}.$$
In other words we have
$$c(v)\leq\sum_{p\in P_1}c_p(t_1).$$
Hence, due to the monotonicity of the backlog costs and the definition of $P_1$ we get that $$\hat{mature}_{t_0}(v)\leq t_1.$$

Consider the scenario in which case 1 occurred. $ALG$ could include the item $v$ in the premature phase of the service $S_C$ and this would still be valid because the sum of the costs of the items in the premature phase would be less than $2\cdot c(r)$ even if $ALG$ included the item $v$ in this phase. This yields a contradiction. 

Now, consider the scenario in which case 2 occurred. The only reason for $ALG$ to choose the item $v'$ to be included in the premature backlog phase of the service $S_C$ rather than the item $v$ is that $$\hat{mature}_{t_0}(v')\leq \hat{mature}_{t_0}(v),$$ that is since $ALG$ chose the items to include in the premature phase in non-decreasing order of $\hat{mature}_{t_0}$. Hence,
$$t_2<\hat{mature}_{t_0}(v')\leq \hat{mature}_{t_0}(v)\leq t_1$$ and that is a contradiction.
\end{proof}
Now that we dealt with the requests $P_1$ we deal with the requests $P_2.$
\begin{lemma}
\label{hold_back_multi_item.lemma.type_d_p_2}
For each request $p\in P_2$ we have $$c_p^{ALG}\leq c_p^{OPT}.$$
\end{lemma}
\begin{proof}
Recall that $ALG$ served the request $t$ during the mature backlog phase of the service $S$, which occurred at time $t_1$, i.e. $t_p^{ALG}=t_1$. In case $OPT$ served the request $p$ after $ALG$ served it, i.e. $t_1\leq t_p^{OPT}$ we have $d_p\leq t_1\leq t_p^{OPT}$ and thus due to the monotonicity of the backlog cost we have $c_p^{ALG}\leq c_p^{OPT}.$ Otherwise we are in the case where $OPT$ served the request $p$ before $ALG$ served $p$. This implies $t_p^{OPT}\leq t_0$ because there were no services triggered by $OPT$ between the service $S_C$, which $ALG$ triggered at time $t_0$, and the service $S$, which $ALG$ triggered at time $t_1$ and in which $ALG$ served the request $p$. We have $$t_p^{OPT}\leq t_0<d_p\leq t_1$$ and hence $$c_p^{ALG}=c_p(t_p^{ALG})=c_p(t_1)<c_p(t_0)\leq c_p(t_p^{OPT})=c_p^{OPT}$$
where the first inequality follows from $p\in P_2$ and the second inequality holds due to the monotonicity of the holding costs. 
\end{proof}
Now we will deal with the requests $P_3.$
\begin{lemma}
\label{hold_back_multi_item.lemma.type_d_p_3}
For each request $p\in P_3$ we have $$c_p^{ALG}\leq c_p^{OPT}.$$
\end{lemma}
\begin{proof}
First,
$$t_0<a_p\leq d_p\leq \underbrace{t_1}_{=t_p^{ALG}}\leq t_p^{OPT}$$ where the first inequality is due to $p\in P_3$ and the last inequality is since $a_p\leq t_p^{OPT}$ and since between the services $S_C$ at time $t_0$ and $S$ at time $t_1$ there were no services that $OPT$ triggered, so $OPT$ had to serve the request $p$ after the service $S$. Due to the monotonicity of the backlog costs we have $c_p^{ALG}\leq c_p^{OPT}$. 
\end{proof}
Using what the claims on requests $P_1$, $P_2$ and $P_3$, we can complete the analysis of every item $v$ that was included in the mature backlog phase of the service $S$ and was not served in the service $S_C$.
\begin{lemma} \label{hold_back_multi_item_lemma_tyep_d_no_s_c}
For each item $v\in V_1(S)\setminus V(S_C)$ we have $$\left(\sum_{p\in B_1^v(S)}c_p^{ALG}\right)-c(v)\leq \sum_{p\in B_1^v(S)}c_p^{OPT}.$$
In other words, the sum of the surplus backlog costs for the item $v$ is upper bounded by the sum of costs $OPT$ paid for serving the requests that $ALG$ served in the mature backlog phase for the item $v$.
\end{lemma}
\begin{proof}
We have 
\begin{align*}
    \left(\sum_{p\in B_1^v(S)}c_p^{ALG}\right)-c(v) 
    & =\sum_{p\in P_3}c_p^{ALG}+\sum_{p\in P_2}c_p^{ALG}+\sum_{p\in P_1}c_p^{ALG}-c(v)\\
    & < \sum_{p\in P_3}c_p^{ALG}+\sum_{p\in P_2}c_p^{ALG}\\
    & \leq \sum_{p\in P_3}c_p^{OPT}+\sum_{p\in P_2}c_p^{OPT} \\
    & \leq \sum_{p\in P_3}c_p^{OPT}+\sum_{p\in P_2}c_p^{OPT}+\sum_{p\in P_1}c_p^{OPT} 
    = \sum_{p\in B_1^v(S)}c_p^{OPT}
\end{align*}
where the first inequality follows from lemma \ref{hold_back_multi_item.lemma.type_d_p_1}. The second inequality follows from summing for all the requests $p\in P_2$ (and using lemma \ref{hold_back_multi_item.lemma.type_d_p_2}) and summing for all the requests $p\in P_3$ (and using lemma \ref{hold_back_multi_item.lemma.type_d_p_3}).
\end{proof}
We now prove lemma \ref{hold_back_multi_item.lemma.type_abd} for the case $S\in\Psi_D$.

\begin{proof}[Proof of Lemma \ref{hold_back_multi_item.lemma.type_abd} (case $S\in\Psi_D$)]
We have
\begin{align*}
c(r) & = c(r)+\sum_{v\in V_1(S)\setminus V(S_C)}c(v)-\sum_{v\in V_1(S)\setminus V(S_C)}c(v) \\
& \leq c(r)+\sum_{v\in V_1(S)}c(v)-\sum_{v\in V_1(S)\setminus V(S_C)}c(v) \\
& =\sum_{p\in B_1(S)}c_p^{ALG}-\sum_{v\in V_1(S)\setminus V(S_C)}c(v) \\
& = \sum_{v\in V_1(S)}\sum_{p\in B_1^v(S)}c_p^{ALG}-\sum_{v\in V_1(S)\setminus V(S_C)}c(v) \\ 
& = \sum_{v\in V_1(S)\cap V(S_C)}\sum_{p\in B_1^v(S)}c_p^{ALG}+\sum_{v\in V_1(S)\setminus V(S_C)}\left(\sum_{p\in B_1^v(S)}c_p^{ALG}\right)-c(v) \\
& \leq \sum_{v\in V_1(S)\cap V(S_C)}\sum_{p\in B_1^v(S)}c_p^{OPT}+\sum_{v\in V_1(S)\setminus V(S_C)}\sum_{p\in B_1^v(S)}c_p^{OPT} \\
& = \sum_{v\in V_1(S)}\sum_{p\in B_1^v(S)}c_p^{OPT} 
= \sum_{p\in B_1(S)}c_p^{OPT} 
\leq \sum_{p\in P(S)}c_p^{OPT}
\end{align*}
where the second equality is since the service $S$ was triggered when the variable $b(r)$, which equals to the sum of the surplus backlog costs for the mature items, reached the value of $c(r)$ and the variable $b(v)$ for each item $v\in B_1(S)$ had a value of $c(v)$, due to the item $v$ being mature. The second inequality follows by:
\begin{itemize}
\item Summing for all the items $v\in V_1(S)\cap V(S_C)$ and all the requests $p\in B_1^v(S)$, using lemma \ref{hold_back_multi_item_lemma_tyep_d_yes_s_c}.
\item Summing for all the items $v\in V_1(S)\setminus V(S_C)$, using lemma \ref{hold_back_multi_item_lemma_tyep_d_no_s_c}.
\end{itemize}
\end{proof}

\subsection{Analyzing the services of types G and H}
\label{section_multi_item_holdBack_gh}

The goal of this subsection is to prove lemma \ref{hold_back_multi_item.lemma.type_gh}. Fix $v\in V$. Our goal is thus to prove that

$$c(v)\cdot(|\Psi_G^{v}|+|\Psi_H^{v}|) \leq c(v)\cdot |\Psi_{OPT}^{v}| +\sum_{S\in \Psi_G^{v}\cup \Psi_H^{v}}\sum_{p\in P^{v}(S)}c_p^{OPT}.$$

Consider the subsequence of services, which contains only the services $\Psi_G^{v}\cup \Psi_H^{v}$ that $ALG$ triggered and the services $\Psi_{OPT}^v$ that $OPT$ triggered. For the purpose of the following definition, we add one dummy service of $OPT$ that will serve the item $v$ and does not serve any request after all the services of $\Psi_G^{v}\cup \Psi_H^{v}$ and $\Psi_{OPT}^v$ have been triggered. We do not charge $OPT$ for the cost of that dummy service. We partition $\Psi_G^{v}\cup \Psi_H^{v}$ to the following four disjoint sets $\Phi_A,\Phi_B,\Phi_C,\Phi_D$ as follows:

\begin{definition}
\label{definition_partition_phi_alg_multi_item} Firstly,
\begin{itemize}
\item Let $\Phi_A$ be all the services in $\Psi_G^{v}\cup \Psi_H^{v}$ that occurred before the first service of $\Psi_{OPT}^v$.
\end{itemize}
Given two consecutive services $S_i,S_{i+1}\in \Psi_{OPT}^v$ triggered at times $t_i$ and $t_{i+1}$ (where $t_i < t_{i+1}$) let $J_i$ be the first service of $\Psi_G^{v}\cup \Psi_H^{v}$ between $S_i$ and $S_{i+1}$ (triggered at time $y_i$) such that $J_i$ served at least one request in the local holding phase of the item $v$ with virtual deadline with respect to time $y_i$ of at least $t_{i+1}$ \textbf{or} the sum of all the holding costs of the requests $J_i$ served in its local holding phase of $v$ was of at most $c(v)$.
\begin{itemize}
\item Let $\Phi_B$ be union over $i$ of all the services of $\Psi_G^{v}\cup \Psi_H^{v}$ between $S_i$ and $J_i$ (excluding).
    If the service $J_i$ as defined above does not exist, then all the services of $\Psi_G^{v}\cup \Psi_H^{v}$ between $S_i$ and $S_{i+1}$ are in $\Psi_B$.
\item Let $\Phi_C$ be union over $i$ of all the services $J_i$ .
\item Let $\Phi_D$ be union over $i$ of all the services of $\Psi_G^{v}\cup \Psi_H^{v}$ between $J_i$ and $S_{i+1}$ (excluding).
\end{itemize}
\end{definition}

Observe the similarities between this definition above and the partition of $\Psi_{ALG}$ in the analysis of the single item problem to the sets $\Psi_A$,$\Psi_B$,$\Psi_C$,$\Psi_D$.

Note that as in the analysis for the single item problem, we can consider a function that maps each service $S\in\Phi_C$ to the last service of $\Psi_{OPT}^v$ that occurred before it, this function is single-valued and the dummy service of $OPT$ from definition \ref{definition_partition_phi_alg_multi_item} does not have any service in $\Phi_C$ that is mapped to it, as in the analysis of the single item problem. This implies the following observation:

\begin{observation}
\label{hold_back_multi_item.lemma.gh.type_c}
    We have that $$|\Phi_C|\leq |\Psi_{OPT}^v|.$$
\end{observation}
Now we define the following:
\begin{definition}
    We say that a service $S\in(\Psi_G^{v}\cup \Psi_H^{v})\setminus \Phi_C$ is \textbf{covered} if the total cost $OPT$ pays for the requests for the item $v$ that $ALG$ served in $S$ is at least $c(v)$. In other words:
    $$c(v)\leq \sum_{p\in P^v(S)}c_p^{OPT}.$$
\end{definition}

In order to prove lemma \ref{hold_back_multi_item.lemma.type_gh}, we prove that all the services $(\Psi_G^{v}\cup \Psi_H^{v})\setminus \Phi_C$ are covered. 

We begin with the services $\Psi_H^{v}\setminus \Phi_C$. Let $S\in \Psi_H^{v}\setminus \Phi_C$ be such a service. Due to the definition of $\Psi_H^{v}$ we have that the item $v$ was mature in the service $S$ and thus the amount of backlog accumulated by the unsatisfied overdue requests for $v$ was at least $c(v)$ during the service $S$. We note that it could have been even higher (up to an additional backlog cost of at most $c(r)$) due to surplus backlog costs for $v$.

The proof that the service $S$ is covered consists of following the process of subsections \ref{section_holdBack_single_item_a}, \ref{section_holdBack_single_item_b} and \ref{section_holdBack_single_item_d} (depending on the set $\Phi_A$, $\Phi_B$, $\Phi_D$ that $S$ belongs to). The differences between this proof and the proof in subsections \ref{section_holdBack_single_item_a}, \ref{section_holdBack_single_item_b} and \ref{section_holdBack_single_item_d} is that instead of referring to the backlog phase of Algorithm \ref{hold_back_single_item.alg}, we refer to the mature backlog phase of Algorithm \ref{hold_back_multi_item.alg}, instead of referring to the holding phase of Algorithm \ref{hold_back_single_item.alg}, we refer to the local holding phase of the item $v$ of Algorithm \ref{hold_back_multi_item.alg}, and we also use $c=c(v)$.

Moreover, while in the analysis for single item we had that the backlog cost served in the backlog phase was exactly $c$, here we have that the backlog cost served in the mature backlog phase may be higher than $c$ (actually $c(v)$ in our case): it is between $c(v)$ and $c(v)+c(r)$, due to surplus backlog cost for $v$ which is between $0$ and $c(r)$. However, it is easy to see that all the proofs for a single item continue to hold even if the total backlog cost is higher than $c$, rather than equal to $c$. 
Indeed, for example, in the analysis for the services $\Psi_A$ that was done in subsection \ref{section_holdBack_single_item_a} in order to prove lemma \ref{hold_back.lemma.type} (for the case $S\in\Psi_A$), we argued that for each request that $ALG$ served in the backlog phase we have that $OPT$ paid at least the cost that $ALG$ paid for it (lemma \ref{hold_back.lemma.type_a_single_request}) and then we claimed that since $ALG$ paid a total backlog cost of $c$ for the requests it served during the backlog phase, by summing for all these requests we have that $OPT$ paid at least the backlog cost that $ALG$ paid, which is $c$. In other words, the backlog cost of $c$ that $ALG$ paid was used as a lower bound to the cost $OPT$ paid for these requests. Even if we had that $ALG$ paid a backlog cost which is more than $c$ for these requests, rather than equal to $c$ the proof still holds.
A similar situation occurs in the analysis of the requests in $\Psi_D$ in single item (or $\Phi_D$ in our case) where we again rely on the fact that the total backlog cost for requests served in the backlog phase of the service is at least $c$ and it does not necessary need to be equal to $c$ for the proof to work.

Another difference between the proof that the services $\Psi_H^{v}\setminus \Phi_C$ are covered and the proof done in subsections \ref{section_holdBack_single_item_a}, \ref{section_holdBack_single_item_b} and \ref{section_holdBack_single_item_d} is that $ALG$ may serve additional requests for the item $v$ during the global holding phase. This difference 
does not affect the proof that the services $\Psi_H^{v}\setminus \Phi_C$ are covered. We thus have the following Corollary.
\begin{corollary}
\label{hold_back_multi_item.lemma.gh.h_covered}
    All the services $\Psi_H^{v}\setminus \Phi_C$ are covered.
\end{corollary}
Now we deal with the services $\Psi_G^{v}\setminus \Phi_C$ and prove that they are covered. This is harder than the proof regarding the services $\Psi_H^{v}\setminus \Phi_C$ since while the minor differences between the proof that the services $\Psi_H^{v}\setminus \Phi_C$ are covered and the proof done for single item did not affect the proofs given in the analysis for single item, here there is a major difference between the services $\Psi_G^{v}\setminus \Phi_C$ and the services in single item. Specifically, the backlog cost that $ALG$ paid for the item $v$ in the premature backlog phase  (due to the definition of $\Psi_G$) is lower than $c(v)$ (it may be even $0$). This is in contrast to the analysis of a single item, where we relied on the fact that it is equal (or higher) than $c$.

In order to overcome this problem, we begin by observing that this difference does not affect the proof as long as services in $\Phi_B$ are concerned: the analysis done in subsection \ref{section_holdBack_single_item_b} for the services $\Psi_B$ did not use the backlog cost of the service and, in particular, did not rely on the fact that the total backlog cost that $ALG$ pays in the service should be equal (or higher) than $c$. Instead, the holding cost which $ALG$ paid during the holding phase was the cost we used as a lower bound to the cost that $OPT$ paid in order to prove that the service is covered. Therefore, we can follow the proof with the minor differences  discussed above when we proved that the services $\Psi_H^{v}\setminus \Phi_C$ are covered. We thus have the following corollary.
\begin{corollary}
\label{hold_back_multi_item.lemma.gh.be_covered}
    All the services $\Psi_G^{v}\cap \Phi_B$ are covered.
\end{corollary}

Before we can complete the proof of lemma \ref{hold_back_multi_item.lemma.type_gh}, we are left with the task of proving that all the services $\Psi_G^{v}\cap (\Phi_A \cup\Phi_D)$ are covered. This is harder than the proof for the other types of services, since in the analysis in subsections \ref{section_holdBack_single_item_a} and \ref{section_holdBack_single_item_d} we relied on the fact that the total backlog cost that $ALG$ pays in the backlog phase of a service is $c$ (or more). However, by the definition of $\Psi_G^{v}$ the backlog cost $ALG$ pays for the requests for the item $v$ in the premature backlog phase of a service $S\in\Psi_G^{v}$ is less than $c(v)$ rather than equal to $c(v)$ (or greater than $c(v)$). Hence (similar to what we did in lemmas \ref{hold_back.lemma.type_a_single_request} and \ref{hold_back.lemma.type_d_single_request}) for each request for the item $v$ that $ALG$ served in the premature backlog phase of $S$, we have that  $OPT$ pays for the request at least the backlog cost that $ALG$ pays for it in $S$. From that we conclude, by summing up for these requests, similarly to subsections \ref{section_holdBack_single_item_a} and \ref{section_holdBack_single_item_d}, that the total cost $OPT$ pays for these requests is at least the total backlog cost that $ALG$ pays for the item $v$ in the premature phase of $S$. However, since this backlog cost that $ALG$ pays is less than $c(v)$ (actually it may be even $0$) then we cannot infer a lower bound of $c(v)$ for the cost $OPT$ pays for these requests. We need another way to prove the following lemma:
\begin{lemma}
\label{hold_back_multi_item.lemma.gh.ad_covered}
    All the services $\Psi_G^{v}\cap (\Phi_A \cup\Phi_D)$ are covered.
\end{lemma}
Lemma \ref{hold_back_multi_item.lemma.gh.ad_covered} is proved in subsection \ref{section_multi_item_holdBack_gh_ad}.
From Lemma \ref{hold_back_multi_item.lemma.gh.ad_covered} and Corollaries \ref{hold_back_multi_item.lemma.gh.be_covered} and \ref{hold_back_multi_item.lemma.gh.h_covered} we have the following corollary:
\begin{corollary}
\label{hold_back_multi_item.lemma.gh.all_covered}
    All the services $(\Psi_G^{v}\cup \Psi_H^{v})\setminus \Phi_C$ are covered.
\end{corollary}
Now we are ready to prove lemma \ref{hold_back_multi_item.lemma.type_gh}.
\begin{proof}[Proof of Lemma \ref{hold_back_multi_item.lemma.type_gh}]
We have 
\begin{align*}
c(v)\cdot(|\Psi_G^{v}|+|\Psi_H^{v}|) & = c(v)\cdot|\Psi_G^{v}\cup\Psi_H^{v}|\\
& = c(v)\cdot|\Phi_C|+c(v)\cdot|(\Psi_G^{v}\cup\Psi_H^{v})\setminus \Phi_C| \\
& \leq c(v)\cdot |\Psi_{OPT}^{v}| +\sum_{S\in \Psi_G^{v}\cup \Psi_H^{v}}\sum_{p\in P^{v}(S)}c_p^{OPT}
\end{align*}
The first equality is due to $\Psi_G^{v}\cap\Psi_H^{v}=\emptyset$. In the inequality we used Observation \ref{hold_back_multi_item.lemma.gh.type_c} and also by summing for all the services $S\in(\Psi_G^{v}\cup \Psi_H^{v})\setminus \Phi_C$, using Corollary \ref{hold_back_multi_item.lemma.gh.all_covered}.
\end{proof}
\subsection{Analyzing the leftover services of types G}
\label{section_multi_item_holdBack_gh_ad}
The goal of this subsection is to prove lemma \ref{hold_back_multi_item.lemma.gh.ad_covered} from subsection \ref{section_multi_item_holdBack_gh}. Let $v$ be an item and let $S\in \Psi_G^{v}\cap (\Phi_A \cup\Phi_D)$ be a service of $ALG$. We need to prove that
$$c(v)\leq \sum_{p\in P^v(S)}c_p^{OPT}.$$
Let $t_1$ be the time when $ALG$ triggered the service $S$. Let $t_2=\hat{mature}_{t_1}(v)$. Let $P_1$ be the set of requests for the item $v$ which were active when $ALG$ triggered the service $S$ and reached their deadline prior to the service $S$. Let $P_2$ be the set of requests for the item $v$ which were active when $ALG$ triggered the service $S$ and  their deadlines are between $t_1$ and $t_2$ and they have bigger (or equal) backlog costs at time $t_2$ than their holding costs at time $t_1$. Formally:
$$P_1=\{p\in A^v_{t_1}:d_p\leq t_1 \}$$
$$P_2=\{p\in A^v_{t_1}: t_1<d_p\leq t_2 \cap c_p(t_1)\leq c_p(t_2)\}$$
Due to the definition of $\hat{mature}$ we have the following observation:
\begin{observation}
\label{inequality_cost_P1_P2}
We have
$$c(v)\leq \sum_{p\in P_1\cup P_2}c_p(t_2).$$
\end{observation}
Consider the first service that $OPT$ triggered after $ALG$ triggered the service $S$, and let $t_3$ be the time of this service of $OPT$ \footnote{Since $S\in\Phi_A \cup \Phi_D$ then $OPT$ served the item $v$ at least once after $ALG$ triggered the service $S$
and thus $t_3$ is well-defined.}. 
Since $S\in \Psi_G^{v}$ then 
$$\max_{v'\in V_2(S)}\hat{mature}_{t_1}(v')\leq t_3.$$
In particular this means
$$t_1< t_2=\hat{mature}_{t_1}(v)\leq t_3.$$
Consider a request $p\in P_1\cup P_2$. If we had $c(v)<c_p(t_2)$ then by continuous of the backlog cost we have (for small enough $\epsilon>0$) that $c(v)<c_p(t_2-\epsilon)$, which implied that $\hat{mature}_{t_1}(v)<t_2$, that yields a contradiction. This implies the following observation:

\begin{observation}
\label{f_p_t_2_at_most_c_v}
For each request $p\in P_1\cup P_2$ we have $$c_p(t_2)\leq c(v).$$
\end{observation}
It is clear that $ALG$ served all the requests $P_1$ during the premature backlog phase of the service $S$ but not necessary all the requests $P_2$ at the service $S$. 
Nevertheless, in case where the sum of holding costs that $ALG$ paid for requests in the local holding phase of the item $v$ in the service $S$ is at most $c(v)$, all requests $P_2$ are served at $S$. This is since if there would be such unversed request $p\in P_2$ then by observation \ref{f_p_t_2_at_most_c_v} we have that $ALG$ could include the request $p$ in the local holding phase of the item $v$ in the service $S$ while the total holding cost would be of at most $2\cdot c(v)$. This yields the following:
\begin{observation}
\label{hold_back_multi_item.lemma.alg_served_all_p_requests_1}
If the sum of holding costs that $ALG$ paid for requests in the local holding phase of the item $v$ in the service $S$ is at most $c(v)$ then $ALG$ served all the requests $P_2$ during the service $S$. In other words:
$$\sum_{p\in H_1^v(S)}c_p^{ALG}\leq c(v) \implies P_2\subseteq  P^v(S)$$
\end{observation}
Now consider the case 
where there is a request $p$ that $ALG$ served during the local holding phase of the item $v$ in the service $S$ which has a virtual deadline with respect to time $t_1$ of later than $t_2$. This also implies that all the requests $P_2$ are served at $S$. The reason is that $ALG$ chooses the requests to serve in the local holding phase according to their virtual deadlines with respect to the time of the service. Note that from the definition of $P_2$ and the monotonicity of the backlog costs we have that for each request $p\in P_2$ that its virtual deadline with respect to $t_1$ is at most $t_2$ and hence
$\hat{d_p^{t_1}}\leq t_2.$
\begin{observation}
\label{hold_back_multi_item.lemma.alg_served_all_p_requests_2}
If there is a request $p$ that $ALG$ served during the local holding phase of the item $v$ in the service $S$ which has a virtual deadline with respect to time $t_1$ of bigger than $t_2$ then $ALG$ served all the requests $P_2$ during the service $S$. In other words:
$$ \left(\exists p\in H_1^v(S):t_2<\hat{d_p^{t_1}} \right) \implies P_2\subseteq  P^v(S)$$
\end{observation}
Now we show the following lemma on the requests $P_1\cup P_2$.
The proof idea follows the proof of lemma \ref{hold_back.lemma.type_d_single_request} from subsection \ref{section_holdBack_single_item_d} but for completeness we prove below the full details.
\begin{lemma}
\label{hold_back_multi_item.lemma.multi_item_gh_ad_single_request_type1}
For each request $p\in P_1\cup P_2$ that $OPT$ pays for it a cost of at most $c(v)$ we have that its cost at time $t_2$ is less or equal than the cost $OPT$ paid for it, i.e. $$\forall p\in P_1\cup P_2: c_p^{OPT}\leq c(v) \implies c_p(t_2)\leq c_p^{OPT}.$$
\end{lemma}
\begin{proof}
We have $d_p\leq t_2$ regardless if $p\in P_1$ or $p\in P_2$.

In case $t_2\leq t_p^{OPT}$ we have $d_p\leq t_2\leq t_p^{OPT}$ and thus due to the monotonicity of the backlog cost we have $c_p(t_2)\leq c_p^{OPT}$. Hence, we need to deal with the case $t_p^{OPT}< t_2$.
Due to $S\in \Psi_G^{v}$ and the definition of $t_2$ we have that $OPT$ did not trigger any service between the service $S$ (which $ALG$ triggered) and time $t_2$. Therefore, we have $t_p^{OPT}\leq t_1$ and $OPT$ served the request $p$ before $ALG$ triggered the service $S$. 
If we had $S\in\Phi_A$ then this would imply that the service $S$ occurred before the first time $OPT$ triggered a service which served the item $v$, in other words the service $S$ occurred before the first service in $\Psi_{OPT}^{v}$, in particular this would imply that $OPT$ served the request $p$ after $ALG$ triggered the service $S$, thus a contradiction. We thus have $S\notin \Phi_A$. Recall that $S\in \Phi_A\cup \Phi_D$. Therefore, we have $S\in \Phi_D$.

Let $S_C\in\Phi_C$ be the furthest service in $\Phi_C$ that was triggered by $ALG$ before the service $S$ and let $t_0$ be the time when $ALG$ triggered the service $S_C$.

Since $S\in\Phi_D$ we know that the item $v$ was not served by $OPT$ between the two services $S_C$ and $S$. 

Recall that we also have that $OPT$ served the request $p$ before $ALG$ triggered the service $S$. We therefore must have that $OPT$ served the request $p$ before the service $S_C$, i.e. $t_p^{OPT}\leq t_0$. We note that there may be services of $OPT$ which have been triggered between the two services $S_C$ and $S$ in which the item $v$ was not served, but this does not disturb the proof because $OPT$ could not serve the request $p$ for the item $v$ in those services.


We have 
\begin{equation}
\label{t_0d_p_t_2}
t_0<d_p\leq t_2 
\end{equation}
where the first inequality is because if we had $d_p\leq t_0$ then $ALG$ would have served the request $p$ in the mature backlog phase or premature backlog phase of the service $S_C$ so the request $p$ would not have been active in $ALG$ at time $t_1$ in the service $S$. 

We also have 
\begin{equation} \label{inequality_times_d_multi_item_g}
t_p^{OPT}\leq t_0< t_1< t_2\leq t_3
\end{equation}
where the second inequality is because $ALG$ triggered the service $S_C$ before it triggered the service $S$ and $ALG$ does not trigger more than one service at the same time.

Next we show that $t_3\leq \hat{d_p^{t_0}}$. 
Since $OPT$ served the request $p$ before the service $S_C$, we have that $p$ was active in $ALG$ at time $t_0$.
Due to the monotonicity of the holding costs and inequality \ref{t_0d_p_t_2} we have $$c_p(t_0)\leq c_p(t_p^{OPT})=c_p^{OPT}\leq c(v).$$
Note that since $S_C\in\Phi_C$ then the service $S_C$ satisfies at least one of the following two conditions:
\begin{enumerate}
    \item There is a request $q$ served by $ALG$ during the local holding phase of the item $v$ in the service $S_C$ with a virtual deadline with respect to time $t_0$ of at least $t_3$, i.e. $\hat{d_q^{t_0}}\geq t_3$.
    \item The sum of the holding costs of the requests $S_C$ served in its local holding phase of $v$ is at most $c(v)$.
\end{enumerate}
If the service $S_C$ satisfied (2) then the request $p$ was active in $ALG$ at time $t_0$ and since $c_p(t_0)\leq c(v)$ then $ALG$ could have included the request $p$ in the local holding phase of the item $v$ in the service $S_C$ and the total holding costs of the requests $S_C$ served in its local holding phase of $v$ would be of at most $2\cdot c(v)$, which implies a contradiction. In other words, the only reason $ALG$ did not include the request $p$ during the local holding phase of $v$ in the service $S_C$ is since (1) above occurred. 

Since $ALG$ chose to include the request $q$ but not the request $p$ in the local holding phase of the item $v$ in the service $S_C$ and this choice is according to their virtual deadlines with respect to time $t_0$, we have that $t_3\leq \hat{d_q^{t_0}}\leq \hat{d_p^{t_0}}$ as needed. This with inequality \ref{inequality_times_d_multi_item_g} implies that 
$t_2\leq \hat{d_p^{t_0}}$. Hence, 
$$c_p(t_2)\leq c_p(\hat{d_p^{t_0}})=c_p(t_0)\leq f(t_p^{OPT})=c_p^{OPT}$$ where the first inequality is due to the inequality above and the monotonicity of the backlog costs, the first equality follows from the definition of virtual deadlines and the second inequality holds due to inequalities \ref{t_0d_p_t_2} and \ref{inequality_times_d_multi_item_g} and the monotonicity of the holding cost.
\end{proof}

Now, We show a lemma regarding the cost of $ALG$ and $OPT$ for requests with early virtual deadlines with respect to time $t_1$. The proof follows the logic of the proof of lemma \ref{hold_back.lemma.type_b_single_request} from subsection \ref{section_holdBack_single_item_b} and appears below for completeness.
\begin{lemma}
\label{hold_back_multi_item.lemma.multi_item_gh_ad_single_request_type2}
For each request $p\in H_1^v(S)$ for the item $v$ that was served during the local holding phase of the item $v$ in the service $S$ that has a virtual deadline with respect to time $t_1$ of at most $t_2$ (i.e. $\hat{d_p^{t_1}}\leq t_2$) we have $$c_p^{ALG}\leq c_p^{OPT}.$$
\end{lemma}
\begin{proof}
We have $t_p^{ALG}=t_1\leq d_p$ since the request $p$ was served by $ALG$ using the service $S$ during the local holding phase of the item $v$.

In the case that $OPT$ served the request $p$ before $ALG$, i.e. $t_p^{OPT}\leq t_1$ we have that $t_p^{OPT}\leq t_1\leq d_p$. The monotonicity of the holding cost implies that $c_p^{ALG}\leq c_p^{OPT}$. We note that this case can only occur if $S\in\Phi_D$ since otherwise $S\in\Phi_A$ then the service $S$, in which $ALG$ served the request $p$, was triggered by $ALG$ before $OPT$ served the item $v$ and the request $p$.

The other case that needs to be considered is when $OPT$ served the request $p$ after $ALG$, i.e. $t_p^{OPT}>t_1$. This certainly happens if $S\in\Phi_A$ but it may also happen if $S\in\Phi_D$. In that case
$$t_1\leq d_p\leq \hat{d_p^{t_1}} \leq t_2\leq t_3\leq t_p^{OPT}$$ 
where the last inequality holds due to the definition of $t_3$ and the fact that $t_1<t_p^{OPT}$. Hence, $$c_p^{ALG}=c_p(t_1)=c_p(\hat{d_p^{t_1}})\leq c_p(t_p^{OPT})=c_p^{OPT}$$ where the second equality is due to the definition of virtual deadlines and the first inequality is due to the monotonicity of the backlog cost.
\end{proof}

Now we are ready to prove lemma \ref{hold_back_multi_item.lemma.gh.ad_covered}. 
\begin{proof}[Proof of Lemma \ref{hold_back_multi_item.lemma.gh.ad_covered}]
If the total holding cost that $ALG$ paid for requests in the local holding phase of the item $v$ in the service $S$ is at least $c(v)$ and all the requests that $ALG$ served in the local holding phase of the item $v$ have virtual deadlines with respect to time $t_1$ of at most $t_2$ then we have
$$c(v)\leq \sum_{p\in H_1^v(S)}c_p^{ALG}\leq \sum_{p\in H_1^v(S)}c_p^{OPT}\leq \sum_{p\in P^v(S)}c_p^{OPT}$$
where in the second inequality we sum for all the requests $p\in H_1^v(S)$, using lemma \ref{hold_back_multi_item.lemma.multi_item_gh_ad_single_request_type2}. 

Otherwise, we can use observation \ref{hold_back_multi_item.lemma.alg_served_all_p_requests_1} or observation \ref{hold_back_multi_item.lemma.alg_served_all_p_requests_2} and get that $P_2\subseteq P^v(S)$. Due to the definition of $P_1$ we also have $P_1\subseteq P^v(S)$. 

Therefore, if there is a request $q\in P_1\cup P_2$ that has $c(v)\leq c_q^{OPT}$ then  $$c(v)\leq c_q^{OPT}\leq\sum_{p\in P_1\cup P_2}c_p^{OPT}\leq \sum_{p\in P^v(S)}c_p^{OPT}.$$
Otherwise, for each request $p\in P_1\cup P_2$ we have $c_p^{OPT}< c(v)$ and thus we can use lemma \ref{hold_back_multi_item.lemma.multi_item_gh_ad_single_request_type1}:
$$c(v)\leq \sum_{p\in P_1\cup P_2}c_p(t_2)\leq \sum_{p\in P_1\cup P_2}c_p^{OPT}\leq \sum_{p\in P^v(S)}c_p^{OPT}$$ where the first inequality is due to observation \ref{inequality_cost_P1_P2} and the second inequality follows from summing for all the requests $p\in P_1\cup P_2$, using lemma \ref{hold_back_multi_item.lemma.multi_item_gh_ad_single_request_type1}.
\end{proof}

\section{Conclusion and Open Problems}
In this paper, we resolved an open problem suggested in \cite{doi:10.1137/1.9781611978322.130} and provided a constant competitive algorithm for  the Joint Replenishment Problem (JRP) with holding and backlog costs with arbitrary (request dependent) functions. Previously, such an algorithm was known only when all requests have the same cost functions.
The natural open problem is to generalize these results to multilevel aggregations. 
There is a line of work on online
algorithms for multilevel trees with hard deadlines or with just backlog costs (and no holding).
Finally, improving competitive ratios for different versions of Online JRP, either by strengthening the upper bounds or proving
lower bounds, is intriguing open questions.

\bibliographystyle{plain}
\bibliography{bib.bib, biblio-2022}
\clearpage

\appendix

\end{document}